%% file: directed.tex
\documentclass[11pt]{article}
\usepackage[utf8]{inputenc}
\usepackage{fullpage}
\usepackage{amsmath,amsfonts,amsthm,amssymb,multirow}
\usepackage{algorithmic}
\usepackage{times}
\usepackage{graphicx}
\usepackage{floatpag}
\usepackage[ruled,vlined,commentsnumbered,titlenotnumbered]{algorithm2e}

\newenvironment{reminder}[1]{\smallskip

\noindent {\bf Reminder of #1 }\em}{\smallskip}

\newtheorem{theorem}{Theorem}[section]

\newtheorem{proposition}{Proposition}
\newtheorem{corollary}{Corollary}[section]

\newtheorem{lemma}{Lemma}[section]
\newtheorem{claim}{Claim}

\makeatletter
\let\c@fconjecture\c@conjecture
\makeatother

\makeatletter
\let\c@fconj\c@conj
\makeatother

\def \eps {\varepsilon}

\newcommand{\ignore}[1]{}

\def\tO{\tilde{O}}

\newcommand{\UndirectedDiameter}[0]{\textsc{UndirectedDiameter}}
\newcommand{\MaxDiameter}[0]{\textsc{MaxDiameter}}
\newcommand{\MinDiameter}[0]{\textsc{MinDiameter}}
\newcommand{\RoundtripDiameter}[0]{\textsc{RoundtripDiameter}}

\newcommand{\UndirectedRadius}[0]{\textsc{UndirectedRadius}}
\newcommand{\SourceRadius}[0]{\textsc{SourceRadius}}
\newcommand{\MaxRadius}[0]{\textsc{MaxRadius}}
\newcommand{\MinRadius}[0]{\textsc{MinRadius}}
\newcommand{\RoundtripRadius}[0]{\textsc{RoundtripRadius}}

\usepackage[dvipsnames,usenames]{color}
\usepackage[colorlinks=true,urlcolor=Blue,citecolor=Green,linkcolor=BrickRed]{hyperref}
\begin{document}

\title{Approximation and Fixed Parameter Subquadratic Algorithms for Radius and Diameter}
\author{Amir Abboud \\ Stanford University \\ \texttt{\small abboud@cs.stanford.edu} \and  Virginia Vassilevska Williams \\ Stanford University \\ \texttt{\small virgi@cs.stanford.edu} \and Joshua Wang \\ Stanford University \\ \texttt{\small jrwang@cs.stanford.edu}}
\date{}
\maketitle

\begin{abstract}
The radius and diameter are fundamental graph parameters. They are defined as the minimum and maximum of the eccentricities in a graph, respectively, where the eccentricity of a vertex is the largest distance from the vertex to another node. 
In directed graphs, there are several versions of these problems. For instance, one may choose to define the eccentricity of a node in terms of the largest distance into the node, out of the node, the sum of the two directions (i.e. roundtrip) and so on. Each of the versions is well-motivated in a variety of applications. All versions of diameter and radius can be solved via solving all-pairs shortest paths (APSP), followed by a fast postprocessing step. Solving APSP, however, on $n$-node graphs requires $\Omega(n^2)$ time even in sparse graphs, as one needs to output $n^2$ distances. 
In this paper, we address the question: when can diameter and radius in sparse graphs be solved in truly subquadratic time, and when is such an algorithm unlikely?

Motivated by known and new negative results on the impossibility of computing these measures exactly in general graphs in truly subquadratic time, under plausible assumptions, we search for \emph{approximation} and \emph{fixed parameter subquadratic} algorithms, and for reasons why they do not exist.

Our results include: 
\begin{itemize}
\item Truly subquadratic approximation algorithms for most of the versions of Diameter and Radius with \emph{optimal} approximation guarantees (given truly subquadratic time), under plausible assumptions.
In particular, there is a $2$-approximation algorithm for directed Radius with one-way distances that runs in $\tilde{O}(m\sqrt{n})$ time, while a $(2-\delta)$-approximation algorithm in $O(n^{2-\eps})$ time is unlikely.
\item On graphs with treewidth $k$, we can solve the problems in $2^{O(k\log{k})}n^{1+o(1)}$ time. We show that these algorithms are near optimal since even a $(3/2-\delta)$-approximation algorithm that runs in time $2^{o(k)}n^{2-\eps}$ would refute the plausible assumptions.
\end{itemize}

%
%

\end{abstract}
\thispagestyle{empty}
\newpage
\setcounter{page}{1}

\newpage

\section{Introduction}

\input{intro}

\section{Subquadratic Approximation Algorithms}
\label{sec:algs}
\input{approx}

\section{Fixed Parameter Subquadratic Algorithms}
\label{sec:fpsub}
\input{fpsub}

\section{Conditional lower bounds}
\label{sec:lb}
\input{roundtripreduction}

\begin{table*}\centering\small
\begin{tabular}{|c | c | c | c|}
\hline
\multicolumn{4}{|c|}{Radius Variants} \\
\hline
Problem & Definition & Upper Bound & HS Conjecture\\
\hline
  \UndirectedRadius{} &
  $\min\limits_c\max\limits_{v} d(c,v)$ &
  $3/2$ in $\tO(m\sqrt{n})$ [\cite{RV13}] &
  $3/2$ [Thm~\ref{thm:radius}] \\
\hline
  \SourceRadius{} &
  $\min\limits_c\max\limits_{v} d(c \to v)$ &
  $2$ in $\tO(m\sqrt{n} \log M)$ [Thm~\ref{thm:source-radius-upper}] &
  $2$ [Thm~\ref{thm:SourceRad}] \\
\hline
  \MaxRadius{} &
  $\min\limits_c\max\limits_{v} \max\{d(c \to v), d(v \to c)\}$ &
  $2$ in $\tO(m)$ [metric]&
  $2$ [Lemma~\ref{lem:max}]
  \\
\hline
  \MinRadius{} &
  $\min\limits_c\max\limits_{v} \min\{d(c \to v), d(v \to c)\}$ & $n$ [Lemma~\ref{thm:finite}]
  &
  $2$ [Lemma~\ref{lem:MinRad}]
  \\
\hline
  \MinRadius{} on DAGs &
  $\min\limits_c\max\limits_{v} \min\{d(c \to v), d(v \to c)\}$ &
  $3$ in $\tO(m\sqrt{n}\log M)$ [Thm~\ref{thm:min-radius-dag-upper}] &
  $2$ [Lemma~\ref{lem:MinRad}]
  \\
\hline
  \RoundtripRadius{} &
  $\min\limits_c\max\limits_{v} \{d(c \to v) + d(v \to c)\}$ &
  $2$ in $\tO(m)$ [metric] &
  $2$ [Thm~\ref{thm:roundtrip}] \\
\hline

\end{tabular}
\caption{Our Bounds for Various Radius Problems}
\label{tab:res-rad}
\end{table*}

\bibliographystyle{plain}
\bibliography{directed}

\appendix

\begin{figure}[t]\begin{center}
\includegraphics[width=.8\columnwidth]{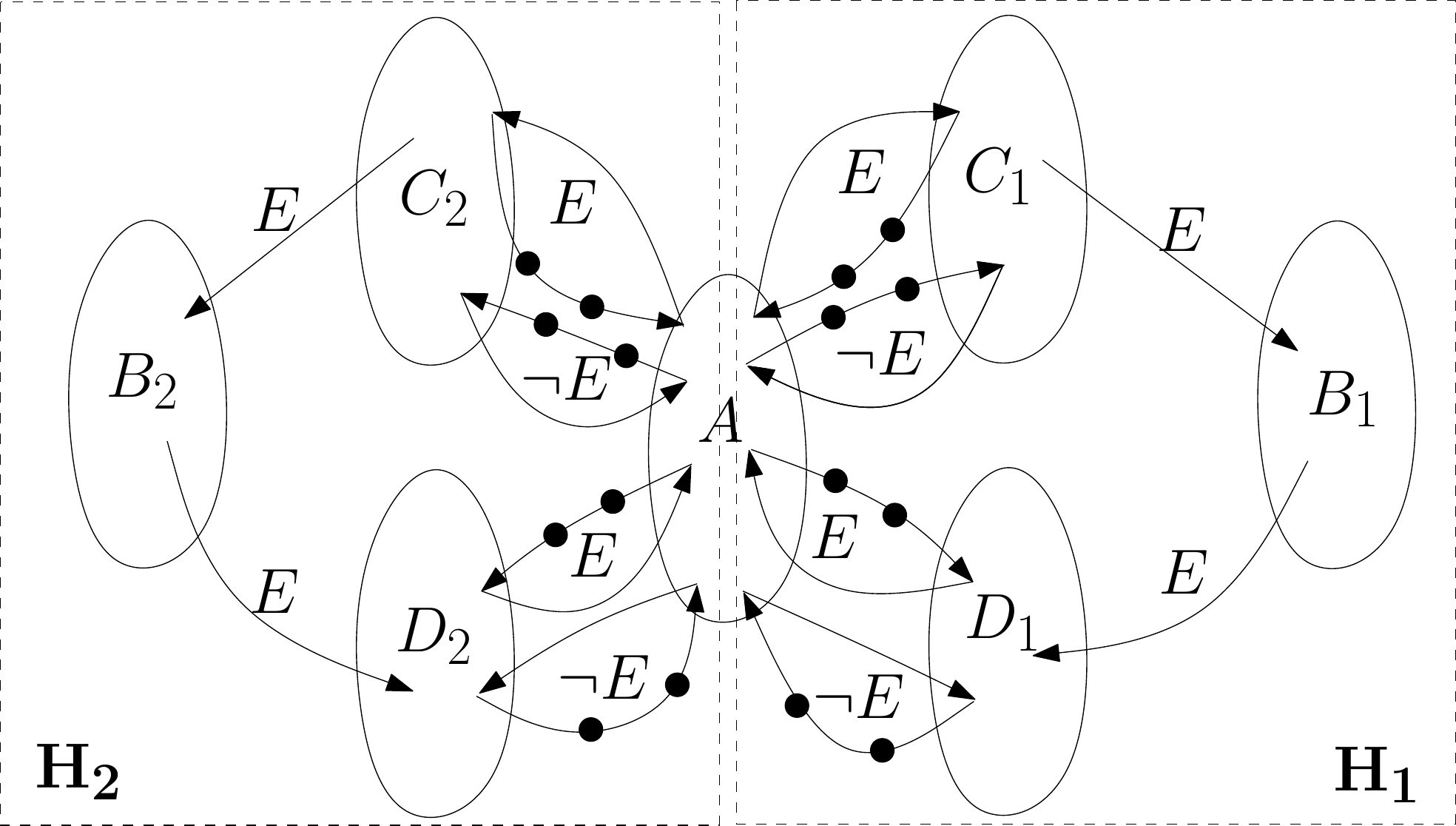} 
\caption{The reduction from HSE to Roundtrip Radius.}\label{fig:reduction}\end{center}
\end{figure}

\section{Subquadratic Reductions}
\label{app:equiv}

\input{equiv}

\section{Missing Algorithms}
\label{app:algs}
\input{upper}

\section{Missing Reductions}
\label{app:lb}

\input{lower}

\end{document}

%% file: intro.tex
Two of the most basic graph parameters are radius and diameter. The diameter of an undirected graph is the largest distance, and the radius is the smallest distance from a node to the furthest node from it. Intuitively, the node that achieves the radius, the so-called center of the graph, is close to all other nodes.
In directed graphs, depending on the application, one may choose to pick whether the center is close in the sense that it has short paths {\em to} other nodes (``source''), {\em from} other nodes (``target''), or even to and then back from other nodes (``roundtrip'').
That is, there are several natural definitions of both radius and diameter for directed graphs. All these variants are well-studied \cite{chung, Hakimi, ChepoiD94,eppstein-planar-jv, aingworth, Corneil01, Chepoi02, Dvir04, BenMoshe, BeKa07, WN08, Yuster10, Chan12, FHW12, WeYu13, RV13, diametersoda14, AGV15, BCH+15} (and many others).
In fact, even estimating the diameter and radius of a network efficiently is useful in practical applications (e.g. the analysis of social networks) and serves as a basic primitive.

Although the problems are very well-studied, essentially the fastest exact algorithms for both Diameter and Radius compute all pairs shortest paths (APSP) and then run a fast postprocessing procedure. Unfortunately, any algorithm for APSP necessarily takes $\Omega(n^2)$ time in $n$-node graphs {\em regardless of the sparsity}, since its output is quadratic. However for Radius and Diameter, whose output is a single integer, it is unclear why $\Omega(n^2)$ time in sparse graphs ($O(n)$ edges) is necessary. In this paper we address the following question.

\begin{center}
{\em When can Diameter and Radius in sparse graphs be solved in $O(n^{2-\eps})$ time for $\eps>0$?}
\end{center}

We provide both algorithms and conditional lower bounds. The study of the above question has a clear practical motivation: quadratic time on real-world graphs is infeasible; ideally, we desired a near-linear time algorithm. There is also a strong theoretical motivation: computing these parameters is one of the most basic graph problems and hence understanding its exact time complexity is of major importance.



In the rest of this paper, we say that a bound is \emph{subquadratic} if it can be bounded by $O(n^{2-\eps})$ for some $\eps>0$, while upper bounds of the form $n^{2-o(1)}$ are only \emph{mildly subquadratic}.

\paragraph{Barriers.}
Recent work has revealed convincing evidence that solving Diameter in subquadratic time might not be possible, even in undirected graphs.
Roditty and Vassilevska W. \cite{RV13} showed that an algorithm that can distinguish between diameter $2$ and $3$ in an undirected sparse graph in subquadratic time refutes the following widely believed conjecture.

\paragraph{The Orthogonal Vectors Conjecture: } There is no $\eps>0$ such that for all $c \geq 1$, there is an algorithm that given two lists of $n$ boolean vectors $A,B \subseteq \{0,1\}^d$ where $d=c\log{n}$ can determine if there is an orthogonal pair $a\in A, b \in B$, in $ O(n^{2-\eps})$ time.

\medskip


The problem in the above conjecture is called the Orthogonal Vectors (OV) problem. The best known algorithm for it runs in mildly subquadratic $n^{2-1/O(\log{(d/\log{n})})}$ time~\cite{AWY15}. Williams~\cite{W04} showed that the OV conjecture is implied by the well-known Strong Exponential Time Hypothesis (SETH) of Impagliazzo, Paturi and Zane~\cite{ipz2,ipz1}. Nowadays many papers base the hardness of problems on 
SETH and the OV conjecture. This holds
both for NP-hard problems (e.g.~\cite{cygan}), as well as problems in P~\cite{PW10,AV14,AVW14,Bring14,BI15,ABV15,BK15}.
%
%

For the Radius problem, the only known barriers to solving the problem exactly are based on other conjectures. Recent work~\cite{AGV15} shows that if the radius of a possibly dense graph can be computed in truly subcubic time, $O(n^{3-\eps})$ for $\eps>0$, then APSP also admits a truly subcubic algorithm. Such an algorithm for APSP has long eluded researchers, and it is often conjectured that it does not exist (e.g. \cite{VW10,AV14,RZ04,Saha14}). For dense graphs the latter result essentially settles the question of computing Radius exactly. For sparse graphs, however, only a much weaker result is known: any $T(m)$ time algorithm for the radius of an $m$-edge graph can be used to find a triangle in an $m$-edge graph in $O(T(m))$ time~\cite{AGV15}. The limit of current techniques for triangle finding is $O(m^{4/3})$~\cite{AlYuZw97} (if the matrix multiplication exponent is $2$), and hence this result gives some reason to believe that obtaining a very fast algorithm for Radius in sparse graphs would be hard. Nevertheless, this result says nothing about the existence of an $O(n^{2-\eps})$ time algorithm.

A natural approach to prove Radius limitations in sparse graphs is to base them on the OV conjecture.
However, such a lower bound has remained elusive~\cite{AGV15,BCH14}.
This is due to the following type mismatch. The OV problem asks for the existence of a pair of vectors with a certain property, just as Diameter asks for the existence of a pair of nodes that are far, i.e. both are of type $\exists x\exists y$.
Meanwhile, Radius asks for the existence of a node such that {\em all} nodes are close, i.e. $\exists x\forall y$.
This quantifier disagreement is the difficulty of proving a lower bound based on OV, and suggests the following natural and plausible variant of the OV conjecture.

\paragraph{The Hitting Set Conjecture: } There is no $\eps>0$ such that for all $c \geq 1$, there is an algorithm that given two lists of $n$ subsets of a universe $U$ of size $c\log{n}$, can decide in $ O(n^{2-\eps})$ time if there is a set in the first list that intersects every set in the second list, i.e. a ``hitting set''.

\medskip

We call the problem in this conjecture the Hitting Set Existence (HSE) problem. An equivalent version of the HSE problem is as follows: given two lists $A,B \subseteq \{0,1\}^d$, determine whether there is a vector $a \in A$  that is not orthogonal to any vector $b \in B$.
The HSE problem can also be solved in mildly subquadratic $n^{2-1/O(\log{(d/\log{n})})}$ time~\cite{AWY15}, where  $d=|U|$.
The HS conjecture is an offline version of folklore conjectured lower bounds on the hardness of classic online problems such as set intersection and partial match studied for instance by Patrascu~\cite{patrascuroditty}.
We discuss these conjectures in Appendix~\ref{app:equiv} and also show that the OV conjecture is implied by the HS conjecture.

With the following theorem, we complete the picture (at least conditionally) for the exact computation of Radius and Diameter in undirected sparse graphs. 
\begin{theorem}
\label{thm:radius}
If for some $\eps>0$, there is an algorithm that can determine if a given undirected, unweighted graph with $n$ nodes and $O(n)$ edges has radius $2$ or $3$ in $O(n^{2-\eps})$ time, then the HS Conjecture is false.
\end{theorem}


%

\paragraph{Overcoming the barriers.} 
The rest of the paper tries to obtain meaningful positive results that overcome the barriers above.
We consider two of the most successful approaches for coping with NP-hard problems: \emph{approximation} and \emph{parameterization}.
In the first approach, we will address questions of the form: what is the smallest constant $c$ such that we can get a $c$-approximation algorithm for Diameter and Radius in directed and undirected graphs in $O(n^{2-\eps})$ time?
In the second approach, we will consider natural parameterizations of Radius and Diameter such as the treewidth of the input graph, and ask whether there is an $O(f(tw)\cdot n^{2-\eps})$ time, or \emph{fixed parameter subquadratic}, algorithm for the problems, and if so, for what functions $f$.

The positive results we obtain in the two parts of our work (corresponding to the two approaches) will use a disjoint set of tools.
However, in both approaches, the upper bounds will be matched (or nearly matched) by lower bounds that are obtained from similar constructions.


\subsection{Approximation algorithms}

In undirected graphs, both Diameter and Radius can be $2$-approximated by a simple linear time algorithm: pick any node and report the largest distance from it.
Aingworth et al.~\cite{aingworth} obtained an $\tilde{O}(n^2+m\sqrt n)$ time almost-$3/2$-approximation algorithm for Diameter and Radius in undirected graphs. Roditty and Vassilevska W.~\cite{RV13} obtained a randomized almost-$3/2$-approximation algorithm with runtime $\tilde{O}(m\sqrt n)$, and Chechik et al.~\cite{diametersoda14} derandomized the algorithm and obtained a genuine $3/2$-approximation algorithm running in time $\tilde{O}(mn^{2/3})$. As previously mentioned,~\cite{RV13} also showed that any $O(n^{2-\eps})$ time algorithm that $(3/2-\delta)$-approximates the diameter (for $\eps,\delta>0$) breaks the OV conjecture (as it would distinguish between graphs of diameter $2$ and $3$). We show that the known approximation algorithms for Radius are also likely tight. An immediate corollary of Theorem~\ref{thm:radius} is:
\begin{corollary} A subquadratic $(3/2-\delta)$-approximation algorithm for {\em Radius}, for some $\delta>0$, refutes the Hitting Set conjecture.\end{corollary}


%

The eccentricity of a node is the largest distance out of it. Diameter is the maximum eccentricity, and radius is the minimum.
Even though both undirected Diameter and Radius can be $3/2$-approximated in subquadratic time, the best known subquadratic algorithm for estimating \emph{all} the eccentricities, by Chechik et al. \cite{diametersoda14}, only gives a $5/3$ approximation.
We show that this result is tight conditioned on the OV conjecture.
 
 \begin{theorem}
 \label{thm:allecc}
 A $(5/3-\delta)$ approximation algorithm for the eccentricities of all nodes in undirected sparse graphs that runs in subquadratic time refutes the Orthogonal Vectors Conjecture.
\end{theorem}

This completes the picture for undirected graphs and we now turn our attention to directed graphs, where much less was known before our work. 
To better highlight the novelty of this work, we will only present our results for Radius on directed graphs (see Table~\ref{tab:res-rad}).
Our results for Diameter can be found in Table~\ref{tab:res-diam}.

\paragraph{One-way distances.}
The first definition of Radius on directed graphs, Source Radius, is the natural extension of the undirected Radius definition: $\min_x \max_v d(x,v)$.
While in undirected graphs a $2$-approximation is trivial, this is no longer the case for directed graphs. In undirected graphs, we can claim for arbitrary $u$ and for all $x$, by the triangle inequality, $d(u,x)\leq d(u,c)+d(c,x)=d(c,u)+d(c,x)\leq 2R$. Since in directed graphs $d(u,c)$ is unrelated to $d(c,u)$, no approximation is guaranteed. 
Even the known $3/2$-approximation algorithms~\cite{RV13,diametersoda14} do not work since for directed graphs all that they can guarantee is that they compute the eccentricity of some node $u$ with either $d(u,c)\leq R/2$ or $d(c,u)\leq R/2$, and in the latter case no approximation can be guaranteed.
Our first algorithmic contribution is a new subquadratic $2$-approximation algorithm for Source Radius overcoming the above issues with a two level sampling approach.

\begin{theorem}
Given a directed unweighted graph on $n$ nodes and $m$ edges, there is an algorithm that outputs $R^*$ such that $R\leq R^*\leq 2R$, and runs in time ${O}(m\sqrt{n}\log^2{n})$.
\end{theorem}

Our algorithm is lightweight and easy to implement. 
Theorem~\ref{thm:radius} implies that a subquadratic algorithm for Source Radius is not likely to have an approximation guarantee better than $3/2$ and makes one wonder whether a $3/2$ guarantee is possible in subquadratic time, as is the case in undirected graphs.
However, using the directed edges we manage to increase the gap in the lower bound construction and prove that the approximation factor of our algorithm is optimal for a subquadratic algorithm under the HS conjecture.

\begin{theorem}
\label{thm:SourceRad}
A $(2-\delta)$-approximation algorithm for Source Radius in sparse graphs that runs in subquadratic time refutes the Hitting Set Conjecture.
\end{theorem}



\paragraph{Roundtrip and longest distances.} 
The roundtrip distance between $u$ and $v$ is the distance from $u$ to $v$ plus the distance from $v$ to $u$, i.e. the sum of both one-way distances.
The Roundtrip Radius of the graph is $\min_{x} \max_v d(x,v)+d(v,x)$.
The Max-distance between $u$ and $v$ is the largest of the two one-way distances.
The Max Radius of the graph is $\min_x\max_v \max\{d(x,v),d(v,x)\}$.

These definitions are natural ways to turn the distances in directed graphs into a metric.
This means that by picking any node as the center we obtain a $2$-approximation near-linear time algorithm for Roundtrip Radius and Max Radius.
Moreover, Cowen and Wagner \cite{CW99,CW00} observed that many of the techniques for approximating distances in undirected graphs can be adapted to handle roundtrip distances, which also led to the roundtrip-spanners of Roditty, Thorup, and Zwick \cite{RTZ08}.
This seems to suggest that these versions of Radius should be more like the undirected version where a $3/2$-approximation is possible in subquadratic time, and not like Source Radius where the $2$ factor is tight.
Quite surprisingly, via a delicate reduction, we were able to obtain a gap of $2$ in the lower bound constructions, and show that anything better than the trivial $2$-approximation is unlikely to run in subquadratic time.

\begin{theorem}
A $(2-\delta)$-approximation algorithm for Roundtrip Radius or Max Radius that runs in $O(m^{2-\eps})$ time, for some $\eps,\delta>0$, refutes the Hitting Set Conjecture.\label{thm:roundtrip}
\end{theorem}

\paragraph{Min Radius.} Finally, we consider a less standard but quite intriguing variant of Radius where distance is the shorter of the two directions.
Formally, we define the Min-eccentricity of a node $c$ to be the maximum over nodes $v$ of $\min\{d(c\to v), d(v \to c)\}$.
The node with minimum Min-eccentricity is the Min-Center of the graph and its Min-eccentricity is the Min-radius.
This directed definition naturally models certain applications. 
For example, in a network representing geographic locations, the Min-center would be the optimal location to place a hospital  since it will allow for the fastest possible medical treatment (either by driving to the hospital or by having an ambulance drive from the hospital to the patient) for any location in the graph.
This is the only directed Radius version without a trivial linear time algorithm on a DAG\footnote{The Max and Roundtrip Radius are infinite on a DAG, and the Source Radius is the eccentricity of the first node in the topological order.}.

Although the problem becomes easy once we compute APSP, it is quite challenging to approximate to within any constant factor without knowing all the distances.
Intuitively, a node with Min-eccentricity $R$ could be very hard to distinguish from nodes that have infinite min-distance to a single node in the graph.
We give a linear time algorithm for this simpler task.

\begin{proposition}
There is an $O(m)$ time algorithm that can check if there is a node in a directed graph with $m$ edges that can reach or be reached from any other node.
Consequently, there is a factor $n$ approximation for Min-Radius in linear time. 
\end{proposition}


Finally, we consider approximation algorithms for Min-Radius on a DAG - which, in our opinion, is the most natural version of the question ``what is the center of a DAG''?
We devise a recursive $3$-approximation subquadratic algorithm for the problem and show that a better than $2$ factor is unlikely.

\begin{theorem}
There is a $3$-approximation algorithm for Min-Radius on $n$ node, $m$ edge DAGs that runs in $O(m\sqrt{n}\log{n})$ time, and a subquadratic $(2-\delta)$ approximation algorithm that runs in subquadratic time on sparse DAGs refutes the Hitting Set Conjecture.
\end{theorem}


\subsection{Fixed Parameter Subquadratic Algorithms}

One of the most active areas of research in theoretical computer science in the past decade is \emph{parameterized} or \emph{multivariate} complexity \cite{DF99,FG06,Nie04}.
The central idea is to study the complexity of an NP-hard problem not only in terms of the input size $n$ but also in terms of an additional natural parameter $k$.
This led to the development of \emph{fixed parameter tractable} algorithms, with running times of the form $f(k)\cdot n^{O(1)}$, for many fundamental problems.

Since quadratic time is a bottleneck in many applications, we propose to treat it as \emph{intractable} as super-polynomial time is traditionally treated. We seek interesting Fixed Parameter Subquadratic algorithms, with running time of the form $O(f(k)\cdot n^{2-\eps})$ for some $\eps>0$.
Unlike classic fixed parameter tractability, this makes sense for problems that are in P before adding a parameter, like Diameter and Radius!

\paragraph{Treewidth.} We will illustrate our approach using the Diameter problem on $n$-node undirected graphs of treewidth $k$.
This is one of the most popular parameterizations of graph problems in the literature on parameterized complexity, and is usually considered when the problem becomes easy on trees \cite{Bod88}.
Note that a folklore algorithm solves Diameter in $\tilde{O}(n)$ on trees: do Dijkstra's from an arbitrary node $u$, and then Dijkstra's from the furthest node $v$ from $u$, report the largest distance found. Since in arbitrary graphs (where the treewidth is $\leq n$) one can solve Diameter in $\tilde{O}(n^2)$ time, a natural conjecture is that the right runtime bound in terms of treewidth is $O(kn)$.
Unfortunately, we observe that the lower bound construction for Diameter rules out any such algorithm. In fact, it shows that in any fixed parameter subquadratic running time for Diameter, the dependence on $k$, the treewidth, must be exponential!

\begin{theorem}
If for some $\eps>0$, there is an algorithm that can distinguish between diameter $2$ and $3$ in an undirected unweighted graph of treewidth $k$ in $2^{o(k)}\cdot n^{2-\eps}$ time, then the Orthogonal Vectors Conjecture is false. 
If such an algorithm exists for Radius, then the Hitting Set Conjecture would be false.
\end{theorem}

This lower bound is quite surprising when contrasted with the near-linear time $(1+\eps)$-approximation algorithm for Diameter in planar graphs of Weimann and Yuster \cite{WeYu13}, since it shows that such a result is unlikely on non-planar graphs of treewidth $\Theta(\log{n})$. 
Furthermore, our lower bound also applies to graphs of {\em pathwidth} $k$.
On the positive side, this bound led us to look for a $2^{O(k)}n^{2-\eps}$ time algorithm for Diameter.

Although computing the treewidth is an NP-hard problem, Bodlaender et al.~\cite{Bod+13} obtained a $2^{O(k)}n$ time algorithm (fixed parameter linear time) that returns a tree decomposition of bag size $O(k)$. 
This gives hope that the known techniques from FPT algorithms will give interesting subquadratic algorithms.
For example, we could apply Courcelle's theorem to solve Diameter in $f(k)\cdot n$ time, for some huge but computable $f(k)$.
Instead, we use a technique that, to our knowledge, was never used for obtaining FPT algorithms for NP-hard problems and obtain a fixed parameter subquadratic algorithm for Diameter and Radius parameterized by treewidth that almost matches our lower bound.

\begin{theorem}
There is an algorithm that solves Diameter and Radius \emph{exactly} in undirected graphs of treewidth $k$ in $2^{O(k \log k)}\cdot n^{1+o(1)}$ time.\label{thm:btw}
\end{theorem}

Closing the small gap in the dependence on $k$ between the $2^{O(k\log k)}\cdot n^{1+o(1)}$ upper bound and the $2^{o(k)}\cdot n^{2-\eps}$ conditional lower bound is a very interesting open question.
In Section~\ref{sec:algs} and the Appendix we also obtain exact algorithms with similar upper bounds for all versions of directed Radius and Diameter that we consider in this work.
We can also compute all the eccentricities of the graph in the same time.

Besides utilizing the tree decomposition to find separators, the main tool in our algorithms is a reduction to an orthogonal range query problem and then using known data structures to answer queries efficiently.
This technique was used by Cabello and Knauer \cite{cabello:knauer:09} to obtain near-linear time algorithms for computing the Wiener index of a fixed treewidth graph\footnote{The Wiener index of a graph is the sum of distances. It can be computed in $O(mn)$ time in general graphs using APSP.}.

The exact running time of our algorithm for Diameter is $O(k^2 n \log^{k-1} n)$ and we believe it can be a practical alternative to known Diameter algorithms when a good bound on the treewidth of the graph is known.
It is known that many real-life networks are tree-like (see \cite{Bod06} and the surveys therein.)

\paragraph{Other parameters.} Perhaps more basic parameterizations for Diameter would be: $D$ - the diameter of the graph, and $\Delta$ - the maximum degree of a node in the graph. 
Unfortunately, the lower bound constructions show that these cases are not fixed parameter subquadratic.
It hard to solve Diameter in subquadratic time even when the diameter is $3$, and there is a simple reduction from Diameter on a sparse graph on $n$ nodes of arbitrary max degree to Diameter on a constant degree graph on $O(n)$ nodes.
The same holds for Radius under the HS conjecture.
Thus, Radius and Diameter are not fixed parameter subquadratic when parameterized by the degree or the diameter of the input graph, unless our conjectures fail.


\paragraph{Related work.} Most related to our parameterized complexity results are known algorithms for Diameter and Radius on special classes of graphs, e.g. \cite{Hakimi, ChepoiD94,eppstein-planar-jv, Corneil01, Chepoi02, Dvir04, BenMoshe, BeKa07, WN08, WeYu13}.
Our two dimensional complexity results, however, show how the complexity changes as the input graph becomes ``more complicated". 
We are not aware of previous negative parameterized complexity results for problems in P.


\subsection{Extensions}

In Section~\ref{app:equiv} we show that there is a \emph{subquadratic equivalence} between the OV problem and the problem of distinguishing between diameter $2$ and $3$ in sparse graphs, in the sense that a subquadratic algorithm for one implies a subquadratic algorithm for the other.
Similarly, there is a subquadratic equivalnce between the HS problem and distinguishing between radius $2$ and $3$ in sparse graphs.
To prove the equivalence we devise new reductions \emph{from} the graph problems \emph{to} OV and HS, via a low-degree high-degree analysis and a hashing trick.
From the mildly subquadratic algorithms for OV and HS \cite{AWY15}, we obtain new mildly subquadratic algorithms for radius and diameter.

\begin{theorem}
There is an algorithm that can decide whether the diameter (or radius) of a given sparse graph is $2$ or $3$, in $O(n^2/2^{c(\sqrt{\log{n}})})$ time, for some $c>0$.
\end{theorem}

This result shows that on sparse $3$-layered graphs, there is a superpolylogarithmic gap between the complexities of diameter and APSP, since there is an unconditional $\Omega(n^2)$ lower bound for APSP.
Such gaps were only known for special classes of graphs (like bounded treewidth graphs), while it is known that the $3$-layered case is typically the hardest for computing distances.

Finally, we demonstrate the potential of the HS conjecture for explaining the hardness of other problems by proving a new conditional lower bound for computing the \emph{median} of the graph.
In undirected graphs, the median is the node $v$ that minimizes the sum of distance to all the other nodes $\sum_{u} d(v,u)$.
Finding the median is equivalent to finding the node with largest \emph{closeness centrality} in the graph \cite{Bavelas50,Bea65,Sab66} - a very important task in network analysis \cite{Hakimi,survey83p1}.
Like Radius, it was known that computing the median of dense weighted graphs in subcubic time refutes the APSP conjecture \cite{AGV15} while no consequences of a subquadratic algorithm in sparse graphs were known.
In stark contrast to Radius, however, Median is known to have a near-linear time $(1+\eps)$ approximation \cite{indyk99,Thorup05,CDPW14a}.
It turns out that the HS conjecture implies that this subquadratic running is impossible if we want to know the median exactly.

\begin{theorem}
\label{thm:median}
A subquadratic algorithm for finding the median of a sparse unweighted undirected graph refutes the Hitting Set Conjecture.
\end{theorem}

%
%
%

%% file: approx.tex
Due to lack of space, we will only present our approximation algorithm for Source Radius here, and refer the reader to Appendix~\ref{app:algs} for the other algorithms.


Although it was trivial to find a $2$-approximation in the
\UndirectedRadius{}, \RoundtripRadius{}, and \MaxRadius{} problems, the
nonsymmetric nature of \SourceRadius{} makes it nontrivial to find an
efficient algorithm that computes a $2$-approximation. Choosing an arbitrary
vertex as before can yield an infinitely bad approximation factor, since it may
not be able to reach all nodes in the graph.

Arbitrary vertices worked before since we could reach the center within $R$
and then any other node within another $R$. Hence a natural attempt is to try
to find a vertex that can reach the center within $R$. Let $Pre(v, \ell)$ be
the set of nodes that can reach $v$ within $\ell$. If $Pre(c, R)$ had many nodes,
we could use a standard hitting set argument to find one of them. This
observation reduces the problem to one where $Pre(c, R)$ is small.

We next make the observation that the center must show up in every $Pre(v, R)$.
If we could find a small $Pre(v, R)$, we could run forward Dijkstra's from
every node there. One way of figuring out which $Pre(v, R)$ are small is to
use the fact that searching for the closest $k$ nodes from a stating node
can be done with a modified Dijkstra in $O(k^2 \log n)$ time.

However, these short Dijkstra's from every node will incur a $\tO(n^2)$ cost
(all of our thresholds are roughly $\sqrt{n}$). Instead, we can be more clever
with how we use our hitting set. With high probability, the hitting set hits
every large $Pre(v, R)$. Hence any $Pre(v, R)$ not hit must be small. At least
one of them must not be hit, since we assumed $Pre(c, R)$ was not.

If we knew the radius, these ideas would give us a running time of
$O(m \sqrt{n} \log n)$. However, doing a binary search for the radius incurs an
additional $(\log Mn)$ factor. To avoid this, we use an idea from the Aingworth
et al. seminal algorithm for a $\frac{2}{3}$-approximation of undirected
diameter; choosing the furthest node from the hitting set simulates locating a
small $Pre(c, R)$ for every $R$ simultaneously.

\begin{theorem}
\label{thm:source-radius-upper}
  There is a $O(m\sqrt{n} \log^2 n)$-time Monte Carlo algorithm that
  approximates \SourceRadius{} on a graph $G$ within a factor of $2$.
\end{theorem}

\begin{proof}
  We claim that algorithm~\ref{alg:apprx-source-radius} has the desired
  properties:

  \begin{algorithm}
  \label{alg:apprx-source-radius}
  \caption{\small ApproximateSourceRadius($G, R$)}
    Sample a hitting set $S_1$ of $O(\sqrt{n} \log n)$ nodes\;
    Run forward Dijkstra's from all $s \in S_1$\;
    Let $w \in V$ maximize $\min_{s \in S_1} d(s, w)$.
    Run a reverse Dijkstra from $w$\;
    Let $S_2$ be the $\sqrt{n}$ closest nodes to $w$.
    Run forward Dijkstra's from all $s \in S_2$\;
    \Return the best source-eccentricity of all nodes in $S_1 \cup S_2$\;
  \end{algorithm}

First, we use a standard argument to claim that for any subset $X$ of $\sqrt{n}$ nodes, 
our random hitting set $S_1$ will intersect $X$ with high probability.

  Now we can prove the claimed approximation guarantee of our algorithm.
   If some $s \in S_1$ can reach the center within $R$, we are done.
  Otherwise, $c$ is more than $R$ away from $S_1$, and hence $w$ is as well.
  Since $S_2$ has $\sqrt{n}$ nodes, it intersects $S_1$ w.h.p. since $S_1$ is a
  hitting set (we want it to hit, for each node, the closest $\sqrt{n}$ nodes
  going backwards). Suppose $v$ is in this intersection. then $d(v, w) > R$. But
  then $S_2$ is defined by how close nodes are to $w$, it must contain all nodes
  that can reach $w$ in less than $R$. This includes $c$, so we are done.

  Now we compute the running time of this algorithm. Running Dijkstras
  from every node in $S_1$ takes $O(m\sqrt{n} \log^2 n)$ time. Running Dijkstra
  from $w$ takes $O(\sqrt{n} \log n)$ time. Finally, running Dijkstras from
  $S_2$ which has $\sqrt{n}$ nodes takes $O(m\sqrt{n} \log n)$ time.
  This completes the proof.
\end{proof}

%% file: fpsub.tex
In this section, we outline our results for diameter and radius on graphs of
small treewidth and provide the key proof ideas. We will focus on the undirected
case, which illustrates the technique. See Appendix~\label{apx:treewidth}
for formal proofs as well as the directed variants.

Our algorithm will actually compute the eccentricity of every node in the graph.
Since diameter is the maximum eccentricity and radius the minimum eccentricity,
we can compute these with only linear postprocessing.

Like many other algorithms, we make use of \emph{portals}: the portals of a
vertex subset $A$ are those nodes of $A$ that have edges going to outside $A$.
Intuitively, finding a vertex subset $A$ which has few portals allows us to divide the
graph into relatively independent pieces. Specifically, if we compute single
source shortest paths from all portals of $A$, we can augment the graph with
weighted edges between portals to account for shortest paths that exit and
re-enter $A$ (or $V \setminus A$). Recursing on augmented graphs yields, for
each node in $A$, the furthest node from it which is also in $A$ (similarly for
$V \setminus A$). If we could compute, for each node in $A$, the furthest node
from it in $V \setminus A$ (and vica versa), we would be done.

But all of these paths pass through some portal. We know the distances from each
node in $A$ to each portal, and from each portal to each node in $V \setminus A$.
We can think of the non-portals of $A$, the portals of $A$, and the nodes in
$V \setminus A$ as forming a three-layered graph. We want to compute, for every
node in the first layer, the furthest node in the third layer (using only
two-hop paths). Note that the second layer only has as many nodes as there were
portals.

As it turns out, this three-layered problem can be written as several max
orthogonal range searching queries. To see this, consider a particular portal
$b$ in the middle layer. When is it the best portal to use to get to a node in
the third layer? If $a$ is a node in the first layer and $c$ a node in the
third, this happens when $d(a, b) + d(b, c) \le d(a, b') + d(b', c)$ for every
other portal $b'$. Using a standard inequality trick, we rearrange to get that
$d(a, b) - d(a, b') \le d(b', c) - d(b, c)$. If we think of each $b'$ as a
coordinate, we can use the right-hand side to transform each vertex $c$ into a
high-dimensional point. The set of $c$ for which $b$ is the best portal, given
an $a$, are exactly those that fall into some orthogonal range. Furthermore,
weighting each vertex by its distance from $b$ allows us to recover the furthest
one when we do a max query. Since these queries can be solved efficiently using
a data structure of Chazelle~\cite{chazelle}, we can solve the three-layered
problem efficiently.

Since the algorithm's running time is highly dependent on the number of portals,
we use a result of Cabello and Knauer~\cite{cabello:knauer:09} which finds a
vertex subset with only $tw(G)$ portals, but is unbalanced (one side may have
$k$ times as many nodes as the other). The resulting algorithm is as follows:
\begin{theorem}
\label{thm:treewidth-undirected}
  There is an algorithm that computes the eccentricity of every vertex
  in an undirected weighted graph $G$ of treewidth at most $k$, in time
  $O(k^2 n \log^{k-1} n)$.
\end{theorem}

%% file: roundtripreduction.tex

In this section we present our lower bound for Roundtrip Radius under the HS conjecture which is a good illustration of the constructions used in all our other reductions. 
All other lower bounds appear in Section~\ref{sec:lb}.

\paragraph{The HSE-Graph.}
All our reductions from HSE will start with the following simple representation of the HSE problem as a ``radius-like" graph problem.

Given an instance $A$, $B$, $U$ of HSE we create the following tripartite graph that we call an ``HSE-graph'' that we will utilize in our reductions. The vertex set is $A\cup B\cup U$ (we overload the notation slightly so that $x$ denotes both a vertex and the corresponding subset in the original instance). The edge set $E$ is as follows: for each $u\in U$ there is an edge to $x\in A\cup B$ if $u\in x$.  The question becomes, is there a node $a\in A$ such that for all $b\in B$ there is a $u\in U$ such that $(a,u),(u,b)\in E$? 
Preprocess the HSE graph as follows. Suppose that there are some $a,a'\in A$ such that $N(a)\subseteq N(a')$ then we can remove $a$ since if $a$ is a hitting set, then so is $a'$. 
Now we can assume that for all $a,a'\in A$, there are $u,u'\in U$ such that $u\in N(a)\setminus N(a'), u'\in N(a')\setminus N(a)$.
We will refer to this as the HSE-graph-problem.

\begin{lemma}
\label{lem:roundtrip}
If for some $\eps>0$, there is an algorithm that can determine if a given directed, unweighted graph with $n$ nodes and $m=O(n)$ edges has roundtrip radius $4$ or $8$ in $O(n^{2-\eps})$ time, then the Hitting Set Conjecture is false.
\end{lemma}

\begin{proof}

We will start from the HSE-graph $G$ with partitions $A',B',U$ and edge set $E$. 
We first build a gadget graph $H$ from $G$ as follows.
$H$ has vertex set $A\cup B\cup C\cup D$ where $A$ is a copy of $A'$, $B$ is a copy of $B'$ and $C$ and $D$ are copies of $U$. For $a\in A'$, let its copy in $A$ also be $a$, and for $b\in B'$ let its copy in $B$ also be $b$. For $u\in U$ let its copies in $C$ and $D$ be $u_C$ and $u_D$, respectively.

If $a\in A',u\in U$, we create a directed $4$-cycle connecting $a\in A$ and $u_C$ and a directed $4$-cycle connecting $a\in A$ and $u_D$ as follows.
If $(a,u)\in E$, then there is an edge from $a$ to $u_C$ and a path of length $3$ directed from $u_C$ to $a$ where the internal nodes of the path are of degree 2 in $H$; additionally, there is an edge from $u_D$ to $a$ and a path of length $3$ from $a$ to $u_D$.
If $(a,u)\notin E$, then the roles of the edges and $3$-paths are reversed. That is, there is a $3$-path from $a$ to $u_C$ and an edge from $u_C$ to $a$ and an edge from $a$ to $u_D$ and a $3$-path from $u_D$ to $a$.
Call the set of internal nodes of all the $3$-paths, $X$.
Each edge $(u,b)$ with $u\in U$, $b\in B$ is represented by two directed edges, $(u_C,b),(b,u_D)$.
Note that any cycle in $H$ has length at least $4$ so that any roundtrip distance within $H$ is also at least $4$.
Now, given $H$ as a gadget, create two copies of $H$, $H_1$ on vertex partitions $(A,B_1,C_1,D_1, X_1)$ and $H_2$ on $(A,B_2,C_2,D_2, X_2)$ so that $H_1$ and $H_2$ are glued at $A$. Call this graph $F$ and see Figure~\ref{fig:reduction} for an illustration.

First suppose that the HSE-instance $G$ was a ``yes'' instance, and there is some $a\in A$ such that for all $b\in B$, there is some $u\in U$ with $(a,u),(b,u)\in E$.
Then we will show that $a$ has roundtrip distance at most $4$ to all nodes in $F$ and hence the roundtrip radius is at most $4$. To see this, first note that by construction, $a$ is on a cycle of length $4$ to every node of $D_1\cup C_1\cup D_2\cup C_2$. 
For any other node $a'\in A$, let $u,u'\in U$ be nodes such that $(a,u),(a',u')\in E, (a,u'),(a',u)\notin E$ (recall such $u,u'$ exist). Then $a\rightarrow u_{C_1}\rightarrow a'\rightarrow u_{D_1}\rightarrow a$ is a directed $4$-cycle in $F$. Finally, for any $b_i\in B_i$ for $i=1,2$, if $u\in U$ is such that $(a,u),(u,b)\in G$,
the following is a directed $4$-cycle in $F$: $a\rightarrow u_{C_i}\rightarrow b_i\rightarrow u_{D_i}\rightarrow a$.

Now suppose that the roundtrip radius of $F$ is $<8$ and we will show that the original graph $G$ must be a ``yes" instance.
We first claim that no node of $F\setminus A$ can be a center.

Case 1. Suppose that some node $u_{C_1}$ is a center (the cases $u_{C_2},u_{D_1},u_{D_2}$ are symmetric).
Then consider the roundtrip shortest path to $u_{C_2}$.
Either the portion of the path  from $u_{C_1}$ to $u_{C_2}$, or the one from  $u_{C_2}$ to $u_{C_1}$, must have length at most $3$. 
Assume, w.l.o.g. that $d(u_{C_1} \to u_{C_2}) \leq 3$, and note that the path must go through $A$. Non of the $3$-paths can be used, since the length would become $>3$, which implies that it must be of the form $u_{C_1} \to a \to u_{C_2}$ for some $a \to A$.
However, by construction, if $(a,u_{C_2})\in E(F)$ then $(a,u)\in E(G)$ and $(u_{C_1},a)\notin E(F)$.
Hence $u_{C_1}$ cannot be a center.

Case 2. Suppose that some node $x$ in $X$ is the center and let $u$ be the closest node in $C$ to $x$.
Note that any roundtrip path from $x$ must go through $u$, which implies that $u$ can only be a better center than $x$.
But by case 1, $u$ cannot be the center and therefore neither can $x$.

Case 3. Now consider any two nodes $b_1\in B_1$, $b_2\in B_2$. By construction, $d(b_1,b_2),d(b_2,b_1)\geq 4$ and hence the roundtrip distance is at least $8$. Hence no node of $B_1\cup B_2$ can be a center.


Hence the center of $F$ is some node $a\in A$. 
Consider the roundtrip distance from $a$ to any $b_1\in B_1$. It is supposed to be at most $7$. Any path from $a$ to $b_1$ that does not go directly from $a$ to some node of $C_1$ to $b_1$ must have length at least $4$. Similarly, any path from $b_1$ to $a$ that does not go directly from $a$ to some node of $D_1$ to $a$ must have length at least $4$. Thus, if the roundtrip radius is $<8$, one of the pieces of the roundtrip path (from $a$ to $b_1$ and from $b_1$ to $a$) must be of length $2$, as otherwise the roundtrip path would be of length at least $8$. 
Hence there is some $u\in U$ for which $(a,u),(u,b)\in E$ and the original graph $G$ is a ``yes'' instance of HSE.

To complete the proof, note that our new graph $F$ has $O(n|U|)$ nodes and $O(n |U|)$ edges. This implies that a subquadratic algorithm for sparse graphs that distinguished between roundtrip radius $4$ and $8$ will solve the HSE problem in $O(n^{2-\eps} \cdot |U|^{2-\eps})$ time, for some $\eps>0$, which refutes the HS conjecture. 
\end{proof}

Finally, we observe that the treewidth (in fact, pathwidth) of the graph in our construction is $O(|U|)$ since by removing all nodes in the $C \cup D$ parts of the graph we are left with a disconnected set of paths.
Thus, an algorithm that can compute Radius on treewidth (or pathwidth) $k$ graphs in $2^{o(k)}\cdot n^{2-\eps}$ can be used to solve the HSE problem where $|U|=\omega(\log{n})$ in $O(n^{2-\eps})$ time, refuting the HS conjecture.


%% file: equiv.tex
\paragraph{Equivalent formulations of the conjectures.}

By very simple reductions, the following problems are equivalent:
\begin{itemize}
\item (Orthogonal Vectors) Given two lists of $n$ vectors in $\{0,1\}^d$ is there an orthogonal pair, one from each list?
\item Given two lists of $n$ sets in $[d]$, is there a pair of sets, one from each list, that are disjoint?
\item Does the product of an $n \times d$ boolean matrix with a $d \times n$ boolean matrix contain any zeros? 
\item (Batch Partial Match) Given a set of $n$ strings of length $d$ over the alphabet $\{0,1,\star\}$, is there a pair that are equal if $\star$ can be treated as any letter?
\end{itemize}

Similarly, the following variants of the above problems are also equivalent:
\begin{itemize}
\item Given two lists of $n$ vectors in $\{0,1\}^d$ is there a vector in the first list that is not orthogonal to any vector in the second list?
\item (HSE) Given two lists of $n$ sets in $[d]$, is there a set in the first list that intersects every set in the second list?
\item Does the product of an $n \times d$ boolean matrix with a $d \times n$ boolean matrix contain a zero in every row? 
\item (No Partial Match) Given two set of $n$ strings of length $d$ over the alphabet $\{0,1,\star\}$, is there a string in the first list that does not match any string from the second list?
\end{itemize}

The first set of problems might look easier because of the alternating quantifiers, which would mean that the HS conjecture should be more likely than the OV conjecture.
However, we show that the opposite is true: a subquadratic algorithm for a problem in the first list will imply a subquadratic algorithm for a problem in the second list.

\begin{proposition}
If the Orthogonal Vectors Conjecture is false, then the Hitting Set Conjecture is also false.
\end{proposition}

\begin{proof}
By Lemma 4.1 in \cite{AWY15} it is known that a $T(n,d)$ algorithm for OV implies a $O(n \cdot T(\sqrt{n},d))$ algorithm for deciding if there is a vector in the first list that is not orthogonal to any vector in the second list.
The latter problem is equivalent to HSE, and the proposition follows by noticing that if $T(n,d)$ can be bounded by $O(n^{2-\eps})$ when $d=\omega(\log{n})$ then we get an $O(n^{2-\eps/2})$ bound for HSE when $d=\omega(\log{n})$.
\end{proof}

Next, we show a reduction from diameter and radius to OV and HSE.
This is the opposite direction of our lower bound proofs, which allows to conclude that the problems are \emph{subquadratic equivalent} and get new mildly subquadratic algorithms for Diameter and Radius on three-layered graphs.

\begin{lemma}
For any $\Delta \in [n]$, if OV can be solved in $T(n,d)$ time, then there is a randomized algorithm that can distinguish between diameter $2$ and $3$ in $n$ node and $m$ edge graphs, w.h.p, in $\tilde{O}(nm/\Delta + T(n,\Delta^2))$ time.\label{lem:algs}
\end{lemma}

\begin{proof}

We use $\Delta$ as a threshold and say that nodes with degree $<\Delta$ are low-degree and otherwise they are high-degree. 
Dijkstra from every high-degree node, and let $D_1$ be the largest distance found.
This step takes $\tilde{O}(m^2/\Delta)$ since there are $O(m/\Delta)$ high-degree nodes.
Assume $a^*,b^*$ is the witness for the diameter, such that $d(a^*,b^*)=D$, and note that if either of them is a high-degree node then $D_1 = D$.

We now handle the case in which the witnesses of the diameter are (both) low-degree, by reduction to OV.

The first idea is to represent the neighborhood of a node $v$ with a vector $\vec{v} \in \{0,1\}^n$, such that $\vec{a},\vec{b}$ are orthogonal iff the distance between $a$ and $b$ is $>2$.
This is straightforward: associate a number in $[n]$ with each node in the graph and let $\vec{v}[j]=1$ if $\{v,j\}\in E$ or if $v=j$, and $\vec{v}[j]=0$ otherwise. 
The only problem is that the dimension of the vector is large, and we will use the fact that the nodes we care about have low-degree to reduce it.

The second idea is to hash each coordinate $j \in [n]$ into a random coordinate $h(j) \in [d]$ where $d=10\Delta^2$.
We can now define new vectors $\vec{v}'$ so that if $\{v,j\} \in E$ we set $\vec{v}'[h(j)] =1$ and set it to $0$ otherwise.
We now claim that if $d(a,b)>2$ then $\vec{a}',\vec{b}'$ are orthogonal with probability at least $2/3$, while if $d(a,b)\leq 2$ then the vectors are \emph{not} orthogonal with probability $1$.

We construct the vectors and call an oracle for OV. 
If an orthogonal pair was found we set $D_2=3$ and otherwise $D_2=2$.
By the above, we have that if $D=2$ then $D_2=2$ with probability $1$, while if $D>2$ then $D_2=3$ with constant probability.
By repeating the above, we can amplify this probability.
This step takes $\tilde{O}(T(n,\Delta^2))$.

Finally, we output $D'=\max\{D_1,D_2\}$. If the diameter is $2$, we always output $D'=2$, and otherwise we output $3$ with very high probability. 
 
\end{proof}

A similar reduction proves the analogous statement for HSE and distinguishing between radius $2$ and $3$.
Since HSE and OV can be solved in $T(n,d)=n^{2-1/O(\log{(d/\log{n})})}$ \cite{AWY15} we get the algorithms for Diameter and Radius in Lemma~\ref{lem:algs}, by setting $\Delta=2^{\Omega(\sqrt{\log{n}})}$.

\section{Lower Bound for Median}
\label{app:median}

In this section we reduce the HSE problem to computing the median of a sparse unweighted undirected graph to prove Theorem~\ref{thm:median}.
An algorithm for Median outputs the quantity $\min_{c \in V} \sum_{v \in V, c \neq v} d(c,v)$.

\begin{reminder}{Theorem~\ref{thm:median}}
A subquadratic algorithm for finding the median of a sparse unweighted undirected graph refutes the Hitting Set Conjecture.
\end{reminder}

\begin{proof}

Given an instance of the HSE problem we construct the corresponding HSE graph $G$ as in Section~\ref{sec:lb}.
We will construct a graph $G'$ from $G$ such that the median of $G'$ tells us whether $G$ is a ``yes" HSE-instance, as follows.
We start by taking $G$ and un-directing all the edges.
Then we add another copy of $U$ to $G'$, call it $U_N$ and denote a copy of $u\in U$ in $U_N$ by $u_N$.
For every $a \in A$ and $u \in U$, if $(a,u) \notin E(G)$ we add the edge $\{a,u_N\}$ to $E(G')$.
Note that $U_N$ is only connected to $A$ and not to $B$ and that $d(a,u)+d(a,u_N)$ is fixed to $4$ for all $a \in A,u \in U$. 
Then, we add a node $x$ and connect every node in $A$ to $x$, and we add $n'=n|U|$ nodes $x_1,\ldots,x_{n'}$ and connect them all to $x$. Call these nodes $X$.
Similarly, we add a node $y$ and connect it to all of $B$, and add $n'$ nodes $y_1,\ldots,y_{n'}$ and connect them all to $y$. Call these nodes $Y$.
Finally, we add a similar gadget and connect it to $U_N$, that is: a node $z$ that is connected to all of $U_N$ and nodes $z_1,\ldots,z_{n'}$ that are connected to $z$. Call these nodes $Z$.

We claim that the median is exactly $M^* = 9n' + 4n + 4|U| + 4$ if $G$ a ``yes" HSE-instance, and $<M^*$ otherwise.

The distance from any node in $A$ to the nodes in $V(G') \setminus B$ is fixed to $2 (n-1) + |U| \cdot (1+3)  + 1 + 2\cdot n'+3+4n' + 2 + 3n' = M^*-2n$ by construction:
The sum of distances to the nodes in $X$ is exactly $1+2n'$, the sum of distances to the nodes in $Y$ is exactly $3+4n'$, and to the nodes in $Z$ it is $2+3n'$.
The sum of distances to the other nodes in $A$ is $2(n-1)$.
The sum of distances to the nodes in $U \cup U_N$ is exactly $4 |U|$, by construction.

For some node $a \in A$,
the sum of distances to the nodes in $B$ (in $G'$) is exactly $2n$ if $a$ could reach every node in $B$ in $G$ (i.e. $a$ was a hitting set), and at least $2n+2$ otherwise.

Thus, the sum of distances from $a \in A$ is exactly $M^*$ iff $a$ is a hitting set, and is at least $M^*+2$ otherwise.

Now, we show that any node that is not in $A$ will have sum of distances greater than $M^*$.
To see this, note that the dominant terms in the sum of distances are the distances to $X,Y,Z$, because of their sizes.
$A$ has distance $2$ to $X$, distance $4$ to $Y$ and distance $3$ to $Z$.
The node $x$ has distance $1$ to $X$ but distance $5$ to $Y$  and $4$ to $Z$, which makes it worse than the nodes of $A$.
The node $y$ has distance $1$ to $Y$ but distance $5$ to $X$ and $6$ to $Z$, and the node $z$ is very far from $Y$.
Similarly, the nodes in $U_N$ are closer by $1$ to $Z$ but further by $1$ from $X,Y$, the nodes in $U$ are closer by $1$ to $Y$ but further by $1$ from $X,Z$, the nodes in $B$ are closer by $2$ to $Y$ but further by $2$ to $X,Z$.
The remaining nodes to consider are the nodes in $X \setminus \{x\}$ (and analogously for $Y,Z$) but those are clearly worse than $x$.

Therefore, the median of $G'$ is $M^*$ iff $G$ was a ``yes" HSE-instance.
To complete the proof node that $G'$ has $O(n|U|)$ nodes and $O(n|U|)$ edges, which implies that a subquadratic algorithm on sparse graphs will solve HSE in subquadratic $\tilde{O}{n^{2-\eps}}$ time for $|U|=\omega(\log{n})$ and refutes the HS conjecture.
\end{proof}

%% file: upper.tex
In this appendix, we cover approximation algorithms for \MinDiameter{} and
\MinRadius{}, as well as giving algorithms for directed diameter and radius for
graphs of small treewidth.

\subsection*{\MinDiameter{}}

We first present an algorithm for \MinDiameter{} on general graphs.

\begin{lemma}
\label{lem:min-diameter-upper}
  Given $\epsilon \ge 0$, there is a $\tO(m n^{1 - \epsilon})$-time algorithm
  that approximates \MinDiameter{} on directed graphs within a factor of
  $n^{\epsilon}$.
\end{lemma}

\begin{proof}
  Suppose that the diameter is realized by the pair of points $(u^*, v^*)$ where
  $d(u^*, v^*) = D$ and $d(v^*, u^*) \ge D$. If $D$ is at most $n^\epsilon$,
  then any edge is a sufficient approximation.

  Consider the case where $D$ is larger than $n^\epsilon$. We choose a hitting
  set $S$ of $\tO(n^{1 - \epsilon})$ nodes that, with high probability, hits the
  middle third of any shortest path longer than $n^\epsilon$. In particular it
  hits the middle third of the shortest path from $u^*$ to $v^*$ at vertex $w$,
  so that $d(u^*, w) \ge \frac{D}{3}$ and $d(w, v^*) \ge \frac{D}{3}$. Notice
  that one of $d(v^*, w)$ or $d(w, u^*)$ must also be at least $\frac{D}{3}$
  long, otherwise $d(v^*, u^*) < D$.

  Hence if we Dijkstra from all nodes in $S$ and return
  $\max_{s \in S, v \in V} \min \{d(s, v), d(v, s)\}$, this yields a
  $n^\epsilon$-approximation. But this takes only $\tO(m n^{1 - \epsilon})$
  time, which completes the proof.
\end{proof}

We get a much better algorithm for \MinDiameter{} on DAGs, since we can use the
topological order of the graph to run a divide-and-conquer.

\begin{theorem}
\label{thm:min-diameter-dag-upper}
  There is a $O(m \log n)$-time algorithm that approximates
  \MinDiameter{} on a DAG $G$ within a factor of $2$.
\end{theorem}

\begin{proof}
  Since $G$ is a DAG, we can run a topological sort in $O(m)$ time and use this
  order to relabel the vertices as $\{0, 1, 2, \ldots, n-1\}$ so that edges run
  from lower- numbered nodes to higher-numbered nodes. Suppose that the diameter
  is realized by the pair of points $(u^*, v^*)$, $u^* < v^*$. There are three
  possible cases:
  \begin{enumerate}
    \item $u^*, v^* < \frac{n}{2}$;
    \item $\frac{n}{2} \le u^*, v^*$;
    \item $u^* < \frac{n}{2} \le v^*$.
  \end{enumerate}

  In case (3), consider node $\frac{n}{2}$, which we denote as $w$.
  $d(u^*, v^*) \le d(u^*, w) + d(w, v^*)$ and so either
  $d(u^*, w) \ge \frac{D}{2}$ or $d(w, u^*) \ge \frac{D}{2}$. Moreover, since
  $u^* \le w \le v^*$, returning \\
  $\max\{\max_{v \le w} d(v, w), \max_{w \le v} d(w, v)\}$ definitely yields a
  $2$-approximation for the diameter. Note that $d(v, w)$ and $d(w, v)$ can be
  computed for all $v$ with a DP in $O(m)$ time.

  Otherwise, if case (3) does not hold, run the algorithm recursively on the
  subgraphs of $G$ induced by the first and last $\frac{n}{2}$ nodes in
  topological order. Building these induced graphs takes $O(m)$ time. There
  are $\log n$ levels of recursion, and each level takes $O(m)$ time: $2^i$
  DPs on $\frac{n}{2^i}$ nodes each where the total number of edges is at most
  $m$. The total time is hence $O(m \log n)$.
\end{proof}

\subsection*{\MinRadius{}}

\MinRadius{} is a difficult problem on general graphs, but it turns out that we
can determine which vertices have a finite min-eccentricity:

\begin{lemma}\label{thm:finite}
  There is a $O(m + n)$-time algorithm that determines which vertices in a
  directed graph $G$ have a finite min-eccentricity.
\end{lemma}

\begin{proof}
  In linear time, we can compute the strongly connected components of $G$. Notice
  that a vertex has a finite min-eccentricity iff its SCC's vertex in the SCC
  graph has a finite min-eccentricity. Hence it suffices to consider the problem
  on DAGs.

  We first compute a topological order of the vertices, which can be done in
  linear time. It suffices for us to determine which nodes can be reached by all
  nodes before them in the topological order, since then we could also compute
  which nodes can reach all nodes \emph{after} them in the topological order by
  symmetry.

  We precompute, for each node, the first node in the topological order it has an
  edge to. This can be done in linear time by taking a minimum over all the
  edges coming out of a node.

  Fix some node $v$. Suppose that every node before $v$ has an edge to a node
  which is before $v$ or is $v$. Then every node before $v$ can reach $v$, since
  we can keep taking edges that do not take us past $v$, and each edge moves us
  forward in the DAG.

  Hence for each node $v$, we will count the number of nodes before $v$ that have
  an edge to a nodes which is before $v$ or is $v$. However, this is easy to do
  with our precomputation. The count is zero for the first node, and the count
  for the $i^{th}$ node is the count for the $(i-1)^{th}$ node plus one (for the
  $(i-1)^{th}$ node itself and minus the number of nodes whose first outward
  edges is to the $i^{th}$ node. Hence we can compute these counts in linear
  time. Nodes can be reached by all nodes before them in topological order iff
  their count is zero, so we can finish in linear time.

  All of our computations took $O(m + n)$ time, as desired. This completes the
  proof.
\end{proof}

Like \MinDiameter{}, \MinRadius{} turns out to be easier on DAGs since we can run a
divide-and-conquer:

\begin{theorem}
\label{thm:min-radius-dag-upper}
  There is a $O(m\sqrt{n}(\log Mn))$-time algorithm that approximates
  \MinRadius{} on a DAG $G$ within a factor of $3$.
\end{theorem}

\begin{proof}
  First we show that there is an algorithm that, given the radius $R$, finds
  a vertex $v \in V$ such that $\epsilon(v) \le 3R$ or guarantees that for all
  vertices $v \in V$, $\epsilon(v) > R$. This algorithm will run in
  $O(m\sqrt{n})$ time. From this claim, we can binary search for $R$ in the
  range $[0, Mn]$, yielding the desired result.

  Since $G$ is a DAG, we can run a topological sort and use this order to
  relabel the vertices as $\{0, 1, 2, \ldots, n-1\}$ so that edges run from
  lower- numbered nodes to higher-numbered nodes. Notice that since $G$ is a
  DAG, if we choose $u, v \in V$ with $u < v$, $d(v, u) = \infty$ so we are
  only concerned with $d(u, v)$. Furthermore, suppose that $d(u, v) > 2R$. We
  claim that the center cannot be in the interval $[u, v]$, since then there
  is a path from $u$ to $v$ through the center with length at most $2R$.

  Algorithm~\ref{alg:apprx-center} uses this observation to return a vertex
  with eccentricity at most $3R$ or guarantees all vertices have eccentricity
  strictly more than $R$.
  \begin{algorithm}
  \label{alg:apprx-center}
  \caption{\small ApproximateCenter($G, R$)}
    Initialize a vector $A$ with $\sqrt{n}$ evenly-spaced vertices, i.e.
      $A[i] = i\frac{n-1}{\sqrt{n}-1}$\;
    \For {$i = 0, 1, \ldots, \sqrt{n}-1$} {
      Use a DP to compute $d(v, A[i])$ and $d(A[i], v)$ for all $v \in V$\;
      \If{$\forall v \in V$, $\min(d(v, A[i]), d(A[i], v)) \le 2R$} {
        \Return $A[i]$\;
      }
    }

    Let $S$ be a stack of vertex intervals, intially empty\;

    \For {$i = 0, 1, \ldots, \sqrt{n}-1$} {
      Let $\ell$ be the topologically-first vertex $v$ such that $d(v, A[i]) > 2R$,
        or $A[i]$ if no vertex satisfies this condition\;
      Let $r$ be the topologically-last vertex $v$ such that $d(A[i], v) > 2R$,
        or $A[i]$ if no vertex satisfies this condition\;
      Suppose the top interval of $S$ is $[a, b]$. If $\ell \le b + 1$, then pop
      $[a, b]$ and push $[a, r]$. Otherwise, just push $[\ell, r]$.
    }

    \For {adjacent vertex intervals $[a, b]$ and $[c, d]$ in $S$} {
      \For {$u \in [b+1, c-1]$} {
        Use a DP to compute $d(v, u)$ and $d(u, v)$ for all $v \in [a, d]$\;
        \If{$\forall v \in [a, d]$, $\min(d(v, u), d(u, v)) \le R$} {
          \Return $u$\;
        }
      }
    }

    \Return all vertices have eccentricity strictly greater than $R$\;
  \end{algorithm}

  First, we will show that Algorithm~\ref{alg:apprx-center} is correct.
  If it returns some node $A[i]$, then every node was within $2R$ of that node
  and hence it does have eccentricity at most $3R$. Otherwise, each node $A[i]$
  has some node $u_i$ that is strictly more than $2R$ away (in the appropriate,
  non-infinite direction). If $u_i < A[i]$, then $[u_i, A[i]]$ cannot contain
  a vertex of eccentricity at most $R$. Similarly, if $A[i] < u_i$, then
  $[A[i], u_i]$ cannot contain a vertex of eccentricity at most $R$. Hence
  every interval of $S$ cannot contain a vertex of eccentricity at most $R$.

  The next phase of the algorithm searches the regions between adjacent
  intervals of $S$ (note that since the first node is in an interval of $S$, as
  well as the last node, all remaining nodes fall between two intervals of $S$).
  Suppose that some $u \in [b+1, c-1]$ can reach all $v \in [a, d]$ in at most
  $R$ distance, either forward or backwards. Then consider the node of $A$
  immediately to its left, $A[i]$. $u$ can reach $A[i]$ (backward) in at most
  $R$ distance. By construction $a$ is either the topologically-first vertex
  $v$ that cannot be reached by $A[i]$ (backwards) in $2R$ distance, or lies
  before that (due to a union with an even earlier region). Hence $u$ can reach
  \emph{all nodes} (backwards) to the left of $a$ with at most $3R$ distance by
  going through $A[i]$. Hence $u$ can reach all nodes before it (backwards)
  using only $3R$ distance. Similarly, it can reach all nodes after it (forwards)
  using only $3R$ distance. Hence $u$ has eccentricity at most $3R$, and is
  valid to return.

  Otherwise, all vertices outside of intervals of $S$ have eccentricity
  strictly more than $R$. But then every vertex has eccentricity more than $R$.
  Hence the final return statement is also correct.

  Next, we analyze the running time of Algorithm~\ref{alg:apprx-center}.
  The first phase of our algorithm computes $\sqrt{n}$ DPs, which take
  $O(m)$ time each. Computing $S$ takes $O(n\sqrt{n})$ time.

  Next, we compute distances for every node not in one of $S$'s intervals.
  In order to bound the running time of this phase, we note two things.
  Firstly, no region between intervals can contain more than $O(\sqrt{n})$
  points since our inital points are all in intervals of $S$ and we chose
  them to be not too far apart. Secondly, any edge only needs to be
  considered for at most two regions between intervals (and only then if
  it lies in some interval of $S$). Since the running time of our DPs
  is linear in the number of edges the DP must consider, our running
  total running time is bounded by $O(m\sqrt{n})$.

  The total running time is hence $O(m\sqrt{n})$, as claimed.

  This completes the proof.
\end{proof}

%
%
%
%

\subsection*{Directed Graphs with Small Treewidth}
\label{apx:treewidth}

In this appendix, we cover formal proofs of our algorithms for diameter and
radius on graphs of small treewidth.


Recall that the \emph{portals} of a vertex subset $A$ are those node which have edges going to outside $A$.
The computation of distances through portals is reduced to orthogonal range searching,
in a similar way to the algorithm of Cabello and Knauer \cite{cabello:knauer:09} for computing the Wiener
index of a graph of treewidth $k$.
Unlike \cite{cabello:knauer:09}, we do not assume that $k$ is a constant.
We remark that our algorithm does not use the treewidth of the graph other than to get separators of size $k$, so we would get the same running time on graphs with separators of size $k$.


We use the following two results from prior work:
\begin{lemma}[\cite{cabello:knauer:09}]
\label{lem:portals}
  Let $k \ge 1$ be a constant. Given a graph $G = (V, E)$ with $n > k + 1$
  vertices and treewidth at most $k$, we can find in $O(2^k n)$ time a subset
  of vertices $S \subseteq V$ such that $S$ has between $\frac{n}{k+1}$ and
  $\frac{nk}{k+1}$ nodes, at most $k$ portals, and adding edges between portals
  of $A$ does not change the treewidth of $G$.
\end{lemma}


\begin{theorem}[\cite{chazelle}]
\label{thm:range-search}
  Consider the range searching for maximum problem: we are given a set $V$ of
  $n$ points in $d$ dimensions and a value function $v : V \to \mathcal{R}$. We
  want to answer queries of the form: given a range of the form
  $q = [a_1, b_1] \times [a_2, b_2] \times \ldots \times [a_d, b_d]$, what is
  $\max_{v \in q} v(p)$?

  On a word RAM, there is a data structure that solves this problem with
  $O(n \log^{d-1} n)$ preprocessing time, $O(n \log^{(d - 1 + \eps} n)$ space
  usage, and $O(\log^{d-1} n)$ query time.
\end{theorem}

We call a directed graph $G = (V, E)$ a {\em three-layered graph} if there is a partition
  of $V$ into $A, B, C$ such that
  $E \subseteq A \times B \cup B \times C$, i.e. all edges go from $A$ to $B$
  or from $B$ to $C$. If $G$ is a three-layered graph, we can also write $G$ as
  $(A, B, C, E)$.
Using the orthogonal range searching data structure in Theorem~\ref{thm:range-search}, we are able to compute
important distances in a three-layered graph. This serves as the key subroutine
for solving diameter and radius on graphs of small treewidth.

\begin{theorem}
\label{thm:three-layered}
  Suppose we have a weighted three-layered graph $G = (A, B, C, E)$.
  Furthermore, suppose that $A$ and $C$ have $O(n)$ nodes while $B$ has only $k$
  nodes. Then we can compute \\ $\max_{c \in C} \min_{b \in B} d(a, b) + d(b, c)$
  for all $a \in A$ in $O(kn \log^{k-2} n)$ time.
\end{theorem}

\begin{proof}
  The key idea is as follows. Focus on some $b \in B$. We will preprocess all of
  the distances between $B$ and $C$ so that when given some $a \in A$, we can
  use its distances to the nodes of $B$ to compute the subset of $C$ whose
  shortest two-hop paths to $a$ go through $b$. However, we don't actually
  compute this set; we instead use our orthogonal range searching data structure
  to return the \emph{furthest} point in the set. This allows us to compute the
  furthest distance any node is from $a$, among nodes that use $b$ as part of
  the shortest path. Looping over all $b \in B$ will then allow us to compute
  the desired quantity.

  Fix $b \in B$. For each $c \in C$ and $b' \in B, b' \neq b$, we compute
  $d(b', c) - d(b, c)$. If we impose an ordering on $B$, this associates a
  $(k - 1)$-dimensional vector with every $c \in C$. Suppose we have some
  $a \in A$ and $c \in C$ where $b$ is the middle vertex in the shortest two-hop
  path from $a$ to $c$. This means that
  $d(a, b) + d(b, c) \le d(a, b') + d(b', c)$ for all other $b' \in B$. We can
  rewrite this as $d(a, b) - d(a, b') \le d(b', c) - d(b, c)$ for all other
  $b' \in B$. In other words, given $a \in A$, the set of $c \in C$ for which
  $b$ is the middle vertex with the shortest two-hop path from $a$ to $c$ are
  those $c$ with vectors that fall in the axis-aligned box given by
  $d(a, b) - d(a, b')$.

  However, by Theorem~\ref{thm:range-search}, there is a data structure that
  does this with only $O(n \log^{k-2})$ preprocessing time and $O(\log^{k-2} n)$
  time per query. Note that the value function we use maps the point
  corresponding to $c$ to $d(b, c)$. Hence the largest weight corresponds to
  the furthest point, and we can compute the distance from $a$ to the furthest
  point by adding $d(a, b)$ to the weight returned.

  We keep one data structure per $b \in B$, and now simply iterate over
  $a \in A$ and $b \in B$. For each $a \in A$, we select the furthest $c$ over
  the two-hop distances computed. This takes $O(kn \log^{k-2} n)$ time.
\end{proof}

We use Theorem~\ref{thm:three-layered} as a subroutine to compute undirected
eccentricities on graphs of small treewidth.

\begin{reminder}{Theorem~\ref{thm:treewidth-undirected}}
  There is an algorithm that computes the eccentricity of every vertex
  in an undirected weighted graph $G$ of treewidth at most $k$, in time
  $O(k^2 n \log^{k-1} n)$.
\end{reminder}

\begin{proof}
  By Lemma~\ref{lem:portals}, we can find $S \subseteq V$ such that $S$
  has between $\frac{n}{k+1}$ and $\frac{nk}{k+1}$ vertices, at most $k$ portals,
  and adding edges between portals of $S$ does not change the treewidth of $G$.
  Finding $S$ takes $O(2^k n)$ time.

  We run Dijkstra from every portal of $S$. Since there are at most $k$ portals,
  this takes $O(k^2 n + k n \log n)$ time. This yields the eccentricity of every
  portal. It remains to compute the eccentricity of non-portals of $S$ and
  vertices in $V \setminus S$.

  The eccentricity of a non-portal of $S$ is either realized by a node in $S$
  or in $V \setminus S$. For the first case, we recurse on $S$
  augmented with weighted edges between portals corresponding to the distances
  between them that we computed via Dijkstra's. Any shortest path between nodes
  of $S$ can be realized by taking a path in this graph; if it goes through at
  least two portals then our added portal-portal edge gives the correct distance.

  To cover the second case, we construct a three-layered graph where $A$
  consists of non-portal nodes of $S$, $B$ consists of portals, and $C$ is
  $V \setminus S$. We add edges from $A$ to $B$ and $B$ to $C$ weighted by the
  Dijkstra distances we computed for the portals. Any shortest path
  between a node in $A$ and a node in $C$ matches the cost of a two-hop path.
  Hence we can use Theorem~\ref{thm:three-layered} to compute, for each
  $a \in A$, the furthest $c \in C$.

  Now, given $a \in A$, we have the furthest distance to any other node in $A$
  and the furthest distance to any node in $V \setminus A$. Hence we can compute
  the eccentricity of $a$ (the max of these two).

  Computing the eccentricities for every node of $V \setminus S$ is identical.
  We recurse on $V \setminus A$ augmented with the portals and weighted edges
  between portals. We also construct a three-layered graph where $A$ is
  $V \setminus S$, $B$ consists of portals, and $C$ consists of non-portal nodes
  of $S$. We again invoke Theorem~\ref{thm:three-layered} on it, and take the
  max of the two computed furthest distances (for each node).

  We now analyze the running time. Invoking Theorem~\ref{thm:three-layered}
  twice takes $O(kn \log^{k-2} n)$ time.  Combining results and constructing
  graphs can be done in $O(k^2 + kn)$ time, which is dominated by
  $O(kn \log^{k-2} n)$.

  We will stop recursing when we have $k^3$ nodes or fewer, which can be
  solved in $O(k^9)$ time by computing all-pairs shortest-paths. We guess that
  the algorithm runs in time $T'(n) = 4k(k+1) n \log^{k-1} n$, and we
  check this inductively. Notice that our case case is covered since $O(k^9)$
  is dominated by $k^5 \log^{k-1} k^3$.
  
  Recall that $\frac{n}{k+1} \le |S| \le \frac{nk}{k+1}$, and that because of
  our base case, $\frac{0.5n}{k+1} \ge k$. The recurrence is
  $T(n) \le kn \log^{k-2} n + T(|S|) + T(n - |S| + k)$.

  \begin{align*}
    T(n) &\le kn \log^{k-2} n + T(|S|) + T(n - |S| + k) \\
         &\le kn \log^{k-2} n
            + 4k(k+1) |S| \log^{k-1} |S|
            + 4k(k+1) (n - |S| + k) \log^{k-1} (n - |S| + k) \\
         &\le kn \log^{k-2} n
            + 4k(k+1)n \log^{k-1} \left( \frac{nk}{k+1} + k \right)
            + 4k^2(k+1) \log^{k-1} \left( \frac{nk}{k+1} + k \right) \\
         &\le kn \log^{k-2} n
            + 4k(k+1)n \log^{k-1} \left( \frac{n(k+0.5)}{k+1} \right)
            + 4k^2(k+1) \log^{k-1} \left( \frac{n(k+0.5)}{k+1} \right) \\
         &\le kn \log^{k-2} n
            + 4k(k+1)n \log^{k-2} n (\log n - \log \frac{k+1}{k+0.5})
            + 4k^2(k+1) \log^{k-1} \left( \frac{n(k+0.5)}{k+1} \right) \\
         &\le kn \log^{k-2} n
            + T'(n)
            - 4k(k+1)n \log^{k-2} n \frac{1}{2k+2}
            + 4k^2(k+1) \log^{k-1} \frac{n(k+0.5)}{k+1} \\
         &\le T'(n)
            - kn \log^{k-2} n
            + 4k^2(k+1) \log^{k-1} n \\
  \end{align*}
  The negative term has at least as much magnitude as the positive term if
  $\frac{n}{\log n} \ge 4k(k+1)$, which is true because $n$ is at least $k^3$.
  Hence our running time is indeed $O(k^2 n \log^{k-1} n)$. This completes the
  proof.
\end{proof}

We can now use our eccentricities to compute the diameter and radius of a graph:

\begin{corollary}
  There are algorithms that compute \UndirectedDiameter{} and
  \UndirectedRadius{} on graphs of treewidth at most $k$ in time
  $O(k^2 n \log^{k-1} n)$.
\end{corollary}

\begin{proof}
  We invoke Theorem~\ref{thm:treewidth-undirected}, observing radius is
  the minimum eccentricity in the graph and diameter is the maximum
  eccentricity in the graph. We can recover both quantities in only $O(n)$
  additional time.
\end{proof}

By noticing that $g(k,n)=k^2 n \log^{k-1} n \leq 2^{2k\log\log{n}} n$ can be upper bounded by $2^{O(k\log{k})}n^{1+o(1)}$ we prove Theorem~\ref{thm:btw} from the Introduction.
This is because when $k \leq \eps \log{n}/\log\log{n}$ we can upper bound $g(n,k)=\tilde{O}(n^{1+\eps})$ and otherwise $k>\eps\log{n}/\log\log{n}$ and therefore $k^2 > \log{n}$ and $\log{k} > \log\log{n}/2$ and we can upper bound $g(n,k) = 2^{O(k\log{k})}\cdot n$.

We now explain simple modifications to Theorem~\ref{thm:treewidth-undirected}
to compute the various directed eccentricities. As before, this means that we
can compute diameter and radius, since they are simply the maximum and minimum
eccentricies. A simple modification gives us max-eccentricities:

\begin{theorem}
  There is an algorithm that computes the max-eccentricity of every vertex in a
  directed weighted graph $G$ of treewidth at most $k$, in time
  $O(k^2 n \log^{k-1} n)$.
\end{theorem}

\begin{proof}
  We make a few modifications to the proof of
  Theorem~\ref{thm:treewidth-undirected}. We must run forward and backward
  Dijkstra's from the portals of $S$ (but this only doubles the running time).
  When recursing, we add directed edges between portals, weighted by the
  distance from the appropriate Dijkstra. We construct twice as many
  three-layered graphs; one weighted by forward distances from $A$ to $B$ and
  $B$ to $C$ and the other will have backward distances from $A$ to $B$ and $B$
  to $C$. The max-eccentricity of a node is just the maximum over its recursive
  value, the distance in the forward three-layered graph, and the distance in
  the backwards three-layered graph.

  The running time analysis is identical.
\end{proof}

\begin{corollary}
  There are algorithms that compute \MaxDiameter{} and \MaxRadius{} on graphs of
  treewidth at most $k$ in time $O(k^2 n \log^{k-1} n)$.
\end{corollary}

Source-eccentricities are also easy:

\begin{theorem}
  There is an algorithm that computes the source-eccentricity of every vertex in
  a directed weighted graph $G$ of treewidth at most $k$, in time
  $O(k^2 n \log^{2k-1} n)$.
\end{theorem}

\begin{proof}
  Again, we make modifications to the proof of
  Theorem~\ref{thm:treewidth-undirected}. We run forward and backward
  Dijkstra's, and recurse with directed edges. We construct three-layered graphs
  weighted by forward distances bewteen $A$ to $B$ and $B$ to $C$.

  The running time analysis is identical.
\end{proof}

\begin{corollary}
  There is an algorithm that computes \SourceRadius{} on graphs of treewidth
  at most $k$ in time $O(k^2 n \log^{k-1} n)$.
\end{corollary}

It may be surprising that we can even solve \MinDiameter{} and \MinRadius{}
efficiently, since they proved difficult in general graphs:

\begin{theorem}
  There is an algorithm that computes the min-eccentricity of every vertex in a
  directed weighted graph $G$ of treewidth at most $k$, in time
  $O(k^2 n \log^{2k-1} n)$.
\end{theorem}

\begin{proof}
  Again, we modify the proof of  Theorem~\ref{thm:treewidth-undirected}. We run
  forward and backward Dijkstra's, and recurse with directed edges. We construct
  three-layered graphs with twice as many nodes in the middle layer. One copy
  will have edges weighted by forward distances from $A$ to $B$ and $B$ to $C$,
  while the other will have edges weighted by backward distances from $A$ to $B$
  and $B$ to $C$. Since distance meausres shortest paths, the furthest distance
  from any $a \in A$ will be the minimum of forward and backward distances to
  some node.

  The running time analysis is almost identical, except invoking
  Theorem~\ref{thm:three-layered} now costs $O(kn \log^{2k-2} n)$ time. Hence
  we pay an additional $O(\log^k n)$ everywhere, to get a running time of
  $O(k^2 n \log^{2k-1} n)$.
\end{proof}

\begin{corollary}
  There are algorithms that compute \MinDiameter{} and \MinRadius{} on graphs of
  treewidth at most $k$ in time $O(k^2 n \log^{2k-1} n)$.
\end{corollary}

We need to do a little more work to get roundtrip-eccentricities. Since the paths
of interest go through two portals, we end up with larger middle layers in our
three-layered graph construction.

\begin{theorem}
  There is an algorithm that computes the roundtrip-eccentricity of every vertex
  in a directed weighted graph $G$ of treewidth at most $k$, in time
  $O(k^2 n \log^{k^2-1} n)$.
\end{theorem}

\begin{proof}
  We again modify the proof of Theorem~\ref{thm:treewidth-undirected}. Run
  forward and backward Dijkstra's, and recurse with directed edges. We construct
  three-layered graphs with $k^2$ nodes in the middle layer, one per \emph{pair}
  of portals. The weights from $A$ to $B$ will correspond to the sum of distance
  to the first portal and distance from the second portal, and weights from $B$
  to $C$ will correspond to the sum of distance from the first portal and
  distance from the second portal.

  Notice that two-hop paths in the three-layered graph now actually correspond
  to roundtrip distances between nodes in $A$ and nodes in $C$, since these
  roundtrips must go through a portal each way.

  The running time analysis is almost identical, except invoking
  Theorem~\ref{thm:three-layered} now costs $O(kn \log^{k^2-2} n)$ time. Hence
  we pay an additional $O(\log^{k^2-k} n)$ everywhere, to get a running time of
  $O(k^2 n \log^{k^2-1} n)$.
\end{proof}

\begin{corollary}
  There are algorithms that compute \RoundtripDiameter{} and \RoundtripRadius{}
  on graphs of treewidth at most $k$ in time $O(k^2 n \log^{k^2-1} n)$.
\end{corollary}

%% file: lower.tex
\subsection{Hitting Set and Radius }

Recall the definition of the HSE-graph from Section~\ref{sec:lb}. 
The reductions in this section will be based on adding gadgets to it.

\paragraph{Undirected Radius.} We are now ready to prove the conditional lower bounds for undirected Radius by simple modifications of the HS-graph.

\begin{reminder}{Theorem~\ref{thm:radius}}
If for some $\eps>0$, there is an algorithm that can determine if a given undirected, unweighted graph with $n$ nodes and $m=O(n)$ edges has radius $2$ or $3$ in $O(n^{2-\eps})$ time, then the Hitting Set Conjecture is false.
\end{reminder}

\begin{proof}
Given an instance $A,B,U$ of the HSE problem, we construct its HSE-graph $G$ as described above.
We will construct an undirected graph $G'$ as follows.
Take $G$ with all its nodes and edges (ignoring the direction of these edges) and add three nodes $x,y,z$ to it.
For each node $a \in A$ add edges $\{a,x\}$ and $\{a,y\}$ to $G'$.
For each node $u \in U$ add an edge $\{u,x\}$ to $G'$.
Finally, add an edge $\{y,z\}$ to $G'$.

We now claim that the radius of $G'$ is $2$ if $G$ is a ``yes" HSE-instance and the radius is at least $3$ otherwise.

First, assume that $G$ is a ``yes" HSE-instance and therefore there is a node $a^*\in A$ such that for every $b \in B$ there is a node $u \in U$ such that both edges $(a^*,u)$ and $(u,b)$ are in $E(G)$ and therefore the edges $\{a^*,u\}$ and $\{u,a^*\}$ are in $E(G')$.
In this case, the distance from $a^*$ to every other node $v$ in $G'$ is at most $2$:
If $v \in B$ then the distance is $2$.
If either $v \in A$ or $v\in U$, then the distance is $2$, via $x$.
If $v=y$ then the distance is $1$ and if $v=z$ then the distance is $2$.

Now, assume that $G$ is a ``no" HSE-instance, which implies that for any node $a\in A$, there is a node $b \in B$ be such that there is no $u \in U$ for which the edges $(a,u)$ and $(u,b)$ are in $E(G)$. 
In this case, there is no path of length $2$ from $a$ to $b$ in $G'$: if the path goes through part $U$ and has length $2$ then it must be of the form $\{a,u\}, \{u,b\}$ for some $u \in U$ which is a contradiction, while if the path goes through $x$ it will have length at least $3$ since the distance from $x$ to any node in $B$ is exactly $2$.
Therefore, if the center of the graph is in $A$, the radius is at least $3$.
On the other hand, if the center of the graph is in $B \cup U \cup \{x\}$ then its distance to $z$ is at least $3$. 
And finally, if the center is $y$ or $z$ then its distance to any node $b \in B$ is at least $3$.
Therefore, the radius of $G'$ is at least $3$.

To complete the proof, note that our new graph $G'$ has $O(n)$ nodes and $O(n |U|)$ edges, and therefore it can be easily turned into a \emph{sparse} graph on $O(n|U|)$ nodes without changing the radius (add $n|U|$ dummy nodes, connect them to a node $d$ and connect $d$ to every node in $A$). This implies that a subquadratic algorithm will solve the HSE problem in $O(n^{2-\eps} \cdot |U|^{2-\eps})$ time, for some $\eps>0$, which refutes the HS conjecture. 

\end{proof}

The following observation about the treewidth (in fact, pathwidth) of the graph in the proof of Theorem~\ref{thm:radius} proves the Radius part of Theorem~\ref{thm:radius}.

\begin{claim}
The Radius instance constructed in the proof of Theorem~\ref{thm:radius} has pathwidth (and therefore treewidth) $O(|U|)$.
\end{claim}

\begin{proof}
Consider the path decomposition $P$ in which there is a bag $B_v$ for every node $v \in A\cup B$ that contains the nodes $B_v = U \cup \{x,y,z\} \cup \{v\}$, and the bags are ordered arbitrarily in a path.
Every edge appears in a bag and all the bags containing a node of $G'$ are connected (they are either a single bag or the whole path).
The sizes of the largest bag is $|U|+4$.
\end{proof}

Thus, an algorithm that can compute Radius on treewidth (or pathwidth) $k$ graphs in $2^{o(k)}\cdot n^{2-\eps}$ can be used to solve the HSE problem where $|U|=\omega(\log{n})$ in $O(n^{2-\eps})$ time, refuting the HS conjecture.

\paragraph{Source Radius.} We now present the reduction to Source Radius which allows us to prove Theorem~\ref{thm:SourceRad}.

\begin{reminder}{Theorem~\ref{thm:SourceRad}}
A $(2-\delta)$-approximation algorithm for Source Radius in sparse graphs that runs in subquadratic time refutes the Hitting Set Conjecture.
\end{reminder}

\begin{proof}

We show how an algorithm that distinguishes between radius $t+1$ and $2t$ on a graph with $O(tn)$ nodes and $O(n |U| + tn)$ edges allows us to solve an HSE instance on two lists of $n$ sets in $U$.
Given an algorithm for Source Radius as in the statement of the theorem, we can
set $t=|U|=\omega(\log{n})$ which allows us to solve HSE in subquadratic time since $(2-\delta)(t+1) <2t$.

Given an instance $A,B,U$ of HSE, we construct the corresponding HSE graph $G$ and then use it to construct our Source Radius instance $G'$ as follows.
Take $G$ with all its nodes and edges and add the following nodes and paths to it to get $G'$.
For each node $b \in B$ add $t-1$ nodes $b_1,\ldots,b_{t-1}$ and add edges so that there is a path $b \to b_1 \to b_2 \to \cdots \to b_{t-1}$.
Similarly, for each node $a \in A$ add $t-2$ nodes $a_1,\ldots,a_{t-2}$ and edges so that there is a path $a_1 \to a_2 \to \cdots \to a_{t-2} \to a$.
Add a node $x$ and connect every node $a \in A$ with an $(a,x)$ edge to $x$, and connect $x$ to the beginning of $a$'s path by adding the edge $(x,a_1)$.

\begin{claim}
The radius of $G'$ is $t+1$ if $G$ is a ``yes" HSE-instance and is at least $2t$ otherwise.
\end{claim}

\begin{proof}
First, assume that $G$ is a ``yes" HSE-instance and therefore there is a node $a^*\in A$ such that for every $b \in B$ there is a node $u \in U$ such that both edges $(a^*,u)$ and $(u,b)$ are in $E(G)$ and therefore these edges are also in $E(G')$.
In this case, the distance from $a^*$ to every other node $v$ in $G'$ is at most $t+1$:
\begin{enumerate} 
\item If $v = b_i$ for some $i \leq t-1$ or $v=b$, then the distance is $2+i \leq t+1$ or $2$.
\item If $v = x$ then the distance is $1$. If $v= a_i$ for some $i \leq t-2$ then the distance is $i+1 \leq t-1$ via $x$. If $v = a \neq a^*$ then the distance is $t$ via $x$.
\item If $v \in U$, then the distance is either $1$ via the edge $(a^*,u)$ or $t+1$ via a path through $x$ to some $a \neq a^*$ and then to $v$.
\end{enumerate}
On the other hand, assume that $G$ is a ``no" HSE-instance, which implies that for any node $a\in A$, there is a node $b \in B$ be such that there is no $u \in U$ for which the edges $(a,u)$ and $(u,b)$ are in $E(G')$. 
In this case, there does not exist a center node $w$ that can reach every other node $v$ within less than $2t$ distance:
\begin{itemize}
\item $w$ cannot be in $A$ since there is no path of length $2t-1$ from $a$ to $b_{t-1}$ in $G'$: the only such paths go through the node $x$ and spend $t$ edges to reach some node $a' \neq a$ and then take a path of length $t+1$ from $a'$ to $b_{t-1}$.
\item Moreover, for any $i \leq t-2$, $w$ cannot be the node $a_i$ since it will have an even larger distance to $b_{t-1}$: it will have to reach $a$ first and then take the path of length $2t+1$ to $b_{t-1}$.
\item $w$ cannot be any node in $U \cup B$ or any $b_i$, since those nodes cannot reach $x$.
\item Finally, $w$ cannot be $x$ since its distance to $b_{t-1}$ is $t-1+t+1=2t$.
\end{itemize}

Therefore, the radius of $G'$ is at least $2t$.
\end{proof}

Our new graph $G'$ has $N=O(tn)$ nodes and $O(n|U| + tn)=O(N)$ edges, which implies that a subquadratic $O(N^{2-\delta})$ algorithm gives a subquadratic algorithm for HSE even when $|U|=\omega(\log{n})$ and refutes the HS conjecture.
\end{proof}

\paragraph{Max Radius.} Now we prove the lower bound for Max Radius. Together with Lemma~\ref{lem:roundtrip}, this proves Theorem~\ref{thm:roundtrip}.

\begin{lemma}
\label{lem:max}
A $(2-\delta)$-approximation algorithm for Max Radius in sparse graphs that runs in subquadratic time refutes the Hitting Set Conjecture.
\end{lemma}

\begin{proof}
The proof proceeds exactly as in the proof of Theorem~\ref{thm:SourceRad}, except that we add the following edges to $G'$:
For every node in $U \cup B$ or any $b_i$ node, add an edge to $x$.

The new edges makes sure that there is a path of length $t+1$ from any node in $G'$ to any node in $A$ and now the same claims hold when we replace one-way distance with max-distance. 
We will outline the differences in the arguments.

If $G$ is a ``yes" HSE-instance, then $G'$ has max-radius at most $t+1$. 
As before, we show a node $a^*$ in $A$ that reaches the other nodes within $t+1$ distance.
Now, however, we also have that any node will reach $a^*$ within distance $t+1$ via $x$, and therefore the max-distance between $a^*$ and the other nodes is $t+1$.

If $G$ is a ``no" HSE-instance, we show that no node of $G'$ can have distance less than $2t$ to the other nodes (this is stronger than having max-eccentricity at least $2t$, since we are ignoring the distances from the other nodes).
Assume for contradiction that there is a node with max-eccentricity $<2t$.
Consider the argument we gave in the proof of Theorem~\ref{thm:SourceRad} and note that it still implies that the center cannot be in $A$ nor $x$ nor any $a_i$ node.
Here, however, we need a different argument for why the center cannot be in $U \cup B$ or any $b_i$:
any such node will have a node in $B$ that is at distance at least $2t$ from it. 
Here we assume that there is no node in $U$ that has edges to every node in $B$ (if such node exists, we check if it has any edges coming from $A$ - if there are we output ``yes" and otherwise we remove the node.)
\end{proof}

\paragraph{Min-Radius.} Finally, we present the lower bounds for Min-Radius. Here, we will have to work harder to make the graph in our construction a DAG.

\begin{lemma}
\label{lem:MinRad}
A $(2-\delta)$-approximation algorithm for Min-Radius on sparse DAGs that runs in subquadratic time refutes the Hitting Set Conjecture.
\end{lemma}

Because we want the input graph to be a DAG, and simultaneously we want the min-distance between any two nodes to be a small constant, we require a special construction.
Given a set $X$ of $n$ nodes $v_1,\ldots,v_n$, it creates a DAG $DG(X)$ with at most $O(n)$ nodes and $O(n\log n)$ edges such that in the topological order of $DG(X)$, $v_i<v_{i+1}$, and for any two nodes of $DG(X)$ $x,y$ where $x<y$ in the topological order, $d(x,y)\leq 2$.

We define $DG(X)$ as follows. WLOG $n$ is a power of $2$, otherwise add enough $(<n)$ new nodes after $v_n$ and grow $n$ until $n$ is a power of $2$.
Using $n-1$ extra nodes, create a complete balanced binary tree $T$ on top of $X$ where the leaves of $T$ are $X$ in the order $v_1,\ldots, v_n$ from left to right. 
Let $r$ be the root of $T$. The edges of $T$ are not added to $DG(X)$ but we will add equivalent directed edges for them; $T$ and $DG(X)$ share the same node set. 

Now, for every node $x$ in $T$, consider the root to $x$ path, and call any node $p$ on the path a $0$ node if the path branches left out of it and a $1$ node otherwise. For every $0$-node $p$ on the $r$-$x$ path, add a directed edge $(x,p)$ to $DG(X)$, and for every $1$-node, add $(p,x)$. Note we have only added $\log n$ edges per node $x$, so the number of edges is $O(n\log n)$. $DG(X)$ is a DAG by construction. Moreover, for any two nodes $x$ and $y$, let their LCA in the tree be $u$. If $u$ is not $x$ and not $y$, then it is a $0$-node for one of them, say $x$, and a $1$-node for the other, $y$. 
Hence, there is a path of length $2$ in $DG(X)$ between any pair of nodes where neither is a descendent of the other in $T$. If $y$ is a descendent of $x$, then the min-distance between then is $1$.

We can generalize the construction above slightly for any integer $t\geq 1$, in a construction $DG_t(X)$ so that for any two nodes $x$ and $y$ with $x<y$ in the topological 
order, where $x$ is not a descendent of $y$ in $T$, their min-distance is $t+1$. To do this, take $DG(X)$ and replace every non-leaf node $\ell$ by a directed path $\ell_1,\ell_2,\ldots,\ell_t$, and replace every in-edge $(x,\ell)$ by $(x,\ell_1)$ and every out-edge by $(\ell,y)$ by $(\ell_t,y)$. We say that $\ell_j$ is a copy of $\ell$ in $T$.
For competeness, call the leaves (the nodes of $X$) copies of themselves and for any leaf $x$, let $x_1=x_t=x$.
For any directed edge $(x,y)$ of $DG(X)$, if $(x,y)$ goes towards the root in $T$, then connect all copies $x_i$ of $x$ to $y_1$, and if $(x,y)$ goes down in $T$, then connect $x_t$ to all copies of $y$. The number of nodes of $DG_t(X)$ is $O(tn)$ and the number of edges is $O(tn\log n)$.


Now, for any nodes of $DG_t(X)$, $x$ and $y$, not descendents of one another, their min-distance is $t+1$ and is attained by taking their LCA, $\ell$ in $T$, taking the edge $(x,\ell_1)$ followed by the path to $\ell_t$ and then the edge $(\ell_t,y)$. If $x$ is a descendent of $y$, then the min-distance is at most $t$.

Here is another way to describe $DG_t(X)$ using binary numbers:

Let us identify each node in $X$ with a number in $[n]$.
We add nodes $v_{i,j}$ for each $i \in [\log{n}]$ and $j \in [2^{i-1}]$ in the form of a binary tree.
Each such node will be connected with a path of length $t-1$: $v_{i,j} \to v_{i,j}^{(1)} \to \cdots \to v_{i,j}^{(t-2)} \to v'_{i,j}$.
These new nodes will be connected to and from $X$ as follows:
for each node $a \in X$, if $a \in [n/2^i (2j-2)+1,\ldots,n/2^i (2j-1)]$, we add an edge $a \to v_{i,j}$, and if $a \in [n/2^i (2j-1)+1,\ldots,n/2^i (2j)]$, we add an edge $v'_{i,j} \to a$.
Note that we only added $O(n \log{n})$ edges, and now there is a path from $a$ to $a'$ for all $a<a'$ in $A$ of length $t+1$, via some $v_{i,j}$.
Then, we also connect the tree nodes with themselves: Let $i,i' \in [\log{n}]$ and $j \in [2^{i-1}], j' \in [2^{i'-1}]$. If $j' \in [2^{i'-i} (2j-2)+1, \ldots,  2^{i'-i} (2j-1)]$, then add edges $v_{i',j'}\to v_{i,j}$, $v'_{i',j'} \to v_{i,j}$ and $v^{(h)}_{i',j'} \to v_{i,j}$ for all  $h \in [t-1]$, and if $j' \in [2^{i'-i} (2j-1)+1, \ldots,  2^{i'-i} (2j)]$, we add $v'_{i,j}\to v_{i',j'}$, $v'_{i,j} \to v'_{i',j'}$ and $v'_{i,j}\to v^{(h)}_{i',j'}$ for all  $h \in [t-1]$.
This completes the construction of this tree structure. 

\begin{proof}


Given an instance $A,B,U$ of HSE, we construct the corresponding HSE graph $G$ and then use it to construct our Min Radius instance $G'$ as follows.

Take $G$ with all its nodes and edges and add the following nodes and paths to it to get $G'$.

For each node $b \in B$ add $t-1$ nodes $b_1,\ldots,b_{t-1}$ and add edges so that there is a path $b \to b_1 \to b_2 \to \cdots \to b_{t-1}$.

Create {\bf two} copies of construction $DG_t(A)$ (sharing $A$). Having two copies rather than one serves to enforce that the center of the graph must be in $A$, and not some extra node of $DG_t(A)$ we added.



Finally, add $t$ nodes $x_1,\ldots,x_t$, connect every node $a \in A$ with an edge $a \to x_1$, then connect every node in $U$ with an edge $x_t \to u$, then add a path $x_1\to \cdots \to x_t$.
Also, add a node $y$ and edges $a \to y$ for every node $a \in A$.

\begin{claim}
The min-radius of $G'$ is $t+1$ if $G$ is a ``yes" HSE-instance and is at least $2t$ otherwise.
\end{claim}

\begin{proof}
First, assume that $G$ is a ``yes" HSE-instance and therefore there is a node $a^*\in A$ such that for every $b \in B$ there is a node $u \in U$ such that both edges $(a^*,u)$ and $(u,b)$ are in $E(G)$ and therefore these edges are also in $E(G')$.
In this case, the min-distance from $a^*$ to every other node $v$ in $G'$ is at most $t+1$:
\begin{enumerate} 
\item If $v = b_i$ for some $i \leq t-1$ or $v=b$, then the distance is $2+i \leq t+1$ or $2$.
\item If $v = x_i$ for some $i \leq t$ then the distance is $i$. If $v=y$ the distance is $1$.
\item If $v \in U$, then the distance is at most $t+1$ via the path $x_1 \to \cdots \to x_{t} \to v$.
\item If $v=a \in A$, then let $i$ be the most significant bit in which the integers $a,a^*$ differ, and note that there must be some $j \in [2^{i-1}]$ such that $a \in [n/2^i (2j-1)+1,\ldots,n/2^i (2j)]$ and $a^* \in [n/2^i (2j-2)+1,\ldots,n/2^i (2j-1)]$, or vice versa. Either way, there is a path of length $t+1$ via $v_{i,j}\to \cdots \to v'_{i,j}$ between the two nodes.

\item Finally, if $v= v^{(h)}_{i,j}$ for some $i \in [\log{n}], j \in[2^{i-1}]$ and $h \in [t-1]$, then let $i^*$ be the largest integer so that the first $i^*$ most significant bits of $a$ and $i$ are the same. As before, this implies that for some $j^*$ there is a path $a \to v_{i^*,j^*} \to \cdots \to v^{(h)}_{i^*,j^*} \to v^{(h)}_{i,j}$, or a path $v^{(h)}_{i,j} \to v^{(h)}_{i^*,j^*} \to \cdots \to v'_{i^*,j^*} \to a$. Thus, the min-distance is at most $t+1$.

\end{enumerate}

On the other hand, assume that $G$ is a ``no" HSE-instance, which implies that for any node $a\in A$, there is a node $b \in B$ be such that there is no $u \in U$ for which the edges $(a,u)$ and $(u,b)$ are in $E(G')$. 
In this case, there does not exist a center node $w$ that can reach or be reached from every other node $v$ within less than $2t$ distance:
\begin{itemize}
\item $w$ cannot be in $A$ since there is no path of length $2t-1$ from $w=a$ to some node $b_{t-1}$ in $G'$ (the one that $a$ cannot reach in $G$): the only such paths either go through the $x_i$-path and spend $t$ edges to reach some node $u \in U$ for which $(a,u)\notin E(G)$, or go through the tree structure, incurring an extra cost of $t$ edges, to reach some node $a' \neq a$, and then take a path of length $t+1$ from $u$ or $a'$ to $b_{t-1}$.
\item $w$ cannot be a node $u \in U$: let $b \in B$ be such that $(u,b)$ is not an edge, then the min-distance between $u$ and $b$ is infinite. Similarly, $w$ cannot be a node $b$ or $b_i$ for a node $b \in B$, since it will have infinite min-distance to the node $b'\neq b$.
\item $w$ cannot be on the $x_i$ path, since it will have infinite min-distance to the node $y$. Thus, $w$ cannot be $y$ as well.
\item  The final case is when $w$ is a node in the tree structure $v^{(h)}_{i,j}$ for some $i,j,h$ in the right range.
This is the more tricky case.
We claim that the min-distance to the copy of this node, in the isomorphic copy of the tree that we added, is infinite. 
To see this, first note that by construction, a path between the two trees must pass through $A$.
Then, note that there is certain threshold $T=n/2^i (2j-1)$, such that for any node $a \in A$ that $v^{(h)}_{i,j}$ can reach, $a> T$, while for any node $a' \in A$ that can reach $v^{(h)}_{i,j}$, $a' \leq T$. 
Since this is also true for the copy of $v^{(h)}_{i,j}$, we conclude that there is no path between these two nodes, and the min-distance is infinite. (In other words, any path from $A$ to $A$ goes through the two copies of $DG(A)$, and since these are copies of the same DAG, it can't be that in one DAG there is a path from $x$ to some $a$, and the other, a path from $a$ to $x$.)
\end{itemize}

Therefore, the radius of $G'$ is at least $2t$.
\end{proof}

Our new graph $G'$ has $O(tn+n)$ nodes and $O(n|U| + tn\log{n})$ edges.
And note that, by construction, $G'$ is a DAG since we do not create any cycles.
It can be turned into a \emph{sparse} graph on $N=O(nt|U|)$ nodes, without changing the min-radius, by adding dummy nodes, all connected to a new node $w$ and connecting $w$ to every node in $A$.
We will choose $t$ large enough such that $(2-\delta)(t+1)<2t$.
Thus, a subquadratic $O(N^{2-\delta})$ time $(2-\delta)$-approximation algorithm for Min-Radius gives a subquadratic algorithm for HSE even when $|U|=\omega(\log{n})$ and refutes the HS conjecture.
\end{proof}

\subsection{Orthogonal Vectors and Diameter}
Similarly to the HSE graph, we define an OV graph $G$ from an OV instance $(A,B)$ where the vectors in $A$ and $B$ have dimension $d=O(\log n)$. $G$ has partitions named $A,B,C$ where $A$ and $B$ (abusing notation slightly) correspond exactly to the sets of vectors $A$ and $B$ of the OV instance, and $C$ is $[d]$. For every $a\in A$, create an edge $(a,c)$ for all $c\in C$ for which $a[c]=1$ and for each $b\in B$ and $c\in C$ for which $b[c]=1$, add $(c,b)$. The OV problem is now to find some $a\in A, b\in B$ such that $b$ is not reachable from $a$.

\paragraph{Min-Diameter.}
\begin{lemma}
\label{lem:MinDiamDAG}
If there is a $(3/2-\eps)$-approximation algorithm for Min-Diameter on a sparse DAG that runs in subquadratic time, then the OV conjecture is false.
\end{lemma}
\begin{proof}

We show how an algorithm that distinguishes between diameter $2$ and $3$ on a DAG allows us to solve the OV problem.
Let $G$ be the OV graph.
We construct a directed graph $G'$ as follows.

Take $G$ (with all its nodes and edges) and add two nodes $x,y$ to it.
Order the nodes in $A$ and add construction $DG(A)$ from the Min-Radius section above.
Similarly, order the nodes in $B$ and add $DG(B)$ and order $C$ and add $DG(C)$.
Finally, add edge $(x,y)$, edges $(a,x),(x,c),(c,y),(y,b)$ for all nodes $a \in DG(A), b \in DG(B), c \in DG(C)$, and for every $a\in DG(A)\setminus A$, add $(a,y)$.

\begin{claim}
The min-diameter of $G'$ is $2$ if $G$ is a NO instance of the OV problem and is at least $3$ otherwise.
\end{claim}

\begin{proof}
First, assume that $G$ is a NO instance and therefore for every pair $a\in A,b \in B$ there is a node $c \in C$ such that both edges $(a,c),(c,b)$ are in $E(G)$ and therefore are also in $E(G')$.
In this case, for all pairs of nodes $u,v \in V(G')$, the min-distance in $G'$ is no more than $2$:
\begin{enumerate} 
\item If $u \in A, v \in B$ then the distance is $2$.
\item If either $u,v \in DG(A)$, or $u,v\in DG(B)$, or $u,v\in DG(C)$, then the distance is at most $2$, by the $DG(\cdot)$ construction.
\item If $u \in DG(A)$ and $v \in DG(C)$ (or vice versa) then the distance is at most $2$ via the path through $x$.
\item If $u \in DG(B)$ and $v \in DG(C)$ (or vice versa) then the distance is at most $2$ via the path through $y$.
\item If $u=x$ then it has distance $1$ to every node in $DG(A) \cup DG(C) \cup \{y\}$ and distance $2$ to every node in $DG(B)$. The $u=y$ case is symmetric, except that the distance to nodes of $DG(A)\setminus A$ is also $1$.
\item If $u\in DG(A)\setminus A$, $v\in DG(B)$ (or vice versa), then the min-distance is $2$ through $y$.
\end{enumerate}
On the other hand, assume that $G$ is a YES instance and let $a\in A, b \in B$ be such that there is no $c \in C$ for which the edges $(a,c)$ and $(c,b)$ are in $E(G)$. 
In this case, there is no path of length $2$ from $a$ to $b$ in $G'$: if the path goes through part $C$ and has length $2$ then it must be of the form $a \to c \to b$ for some $c \in C$ which is a contradiction. If the path goes through $x$ it will have distance at least $3$ since the distance from $x$ to any node in $B$ is exactly $2$.
Finally, if the path goes through $DG(A)\setminus A$, through $y$, then it also has length at least $3$.
Therefore, the min-diameter of $G'$ is at least $3$.
\end{proof}
Our new graph $G'$ has $O(n)$ nodes and $O(n\log{n})$ edges, which implies that a subquadratic algorithm refutes the OV conjecture. 
$G'$ is a DAG by construction: $DG(A), DG(B),DG(C)$ are DAGs and the rest of the comparisons in the topological order are $DG(A)<x<DG(C)<y<DG(B)$.
\end{proof}
\begin{lemma}
\label{lem:MinDiam}
If there is a $(2-\eps)$-approximation algorithm for Min-Diameter on a sparse weighted graph that runs in subquadratic time, then the OV Conjecture is false.
\end{lemma}

\begin{proof}

We show how an algorithm that distinguishes between diameter $t+1$ and $2t$ on a sparse graph allows us to solve OV.
Let $G$ be the OV graph. We construct a directed graph $G'$ as follows.
Take $G$ (with all its nodes and edges) and add three nodes $x,y,z$ to it.
The edges that were present in $G$ will have weight $t/2$ in $G'$.
We connect $x$ to and from $A$ with edges $(a,x),(x,a)$ for every node $a \in A$, with weights $w(a,x)=1,w(x,a)=t$, and then we add edges $(c,x)$ for every node $c \in C$, with weight $w(c,x)=1$.
We connect $y$ to and from $B$ with edges $(b,y),(y,b)$ for every node $b \in B$, with weights $w(b,y)=t, w(y,b)=1$,  and then we add edges $(y,c)$ for every node $c \in C$, with weight $w(y,c)=1$.
Finally, we connect $z$ to and from $C$ with edges $(c,z),(z,c)$ for every node $c \in C$, with weights $w(c,z)=w(z,c)=t/2$, and we add an edge $(y,x)$ with weight $w(y,x)=1$.

\begin{claim}
The min-diameter of $G'$ is $t+1$ if $G$ is a NO instance of OV and is at least $2t$ otherwise.
\end{claim}

\begin{proof}
First, assume that $G$ is a NO instance and therefore for every pair $a\in A,b \in B$ there is a node $c \in C$ such that both edges $(a,c),(c,b)$ are in $E(G)$ and therefore are also in $E(G')$.
In this case, for all pairs of nodes $u,v \in V(G')$, the min-distance in $G'$ is no more than $2$:
\begin{enumerate} 
\item If $u \in A, v \in B$ then the distance from $u$ to $v$ is $2 t/2 = t$.
\item If either $u,v \in A$, or $u,v\in B$, or $u,v\in C$, then the distance is $t+1$, via $x,y, or z$.
\item If $u \in A$ and $v \in C$ (or vice versa) then the distance from $v$ to $u$ (or vice versa) is at most $t+1$ via the path through $x$.
\item If $u \in B$ and $v \in C$ (or vice versa) then the distance from $u$ to $v$ (or vice versa) is at most $t+1$ via the path through $y$.
\item If $u=x$ then it has min-distance $1$ to every node in $A \cup C \cup \{y\}$, distance $t+1$ from every node $b\in B$ via the path $b \to y \to x$, and distance $t/2+1$ from $z$ via the path $z \to c \to x$ for any $c \in C$. The $u=y$ case is symmetric.
\item If $u=z$ then it has min-distance $t/2$ to every node in $C$ and min-distance $t/2+t/2$ to every node in $A \cup B$.
\end{enumerate}
On the other hand, assume that $G$ is a YES instance and let $a\in A, b \in B$ be such that there is no $c \in C$ for which the edges $(a,c)$ and $(c,b)$ are in $E(G)$. 
In this case, there is no path of length less than $2t$ from $a$ to $b$ in $G'$: if the path does not any of the nodes $x,y,z$ it must be of the form $a \to c \to b$ which is a contradiction, while if it uses $x$ it must be of the form $a \to \cdots \to x \to a' \to c \to \cdots \to \cdots b$ which will have length at least $2t$, the case in which the path goes from $b$ to $a$ and uses $y$ is symmetric, and finally, if we use $z$ we also incur an addition weight of $t$ resulting in length at least $2t$.
Therefore, the min-diameter of $G'$ is at least $2t$.
\end{proof}

\end{proof}

%
%
%
\paragraph{Roundtrip Diameter.}

\begin{lemma}
\label{lem:RoundDiam}
If there is a $(3/2-\eps)$-approximation algorithm for Roundtrip-Diameter on a sparse graph that runs in subquadratic time, then the OV conjecture is false.
\end{lemma}

\begin{proof}
Any algorithm distinguishing between roundtrip diameter $4$ and $6$ can distinguish between undirected diameter $2$ and $3$ - take the undirected graph, bi-direct the edges and run the algorithm for roundtrip-diameter.
The latter task cannot be done in subquadratic time \cite{RV13}.
\end{proof}

\subsection{Lower Bound for Estimating All eccenricities}

The following construction shows that the $3/2$ factor that we have for Radius and Diameter (in undirected graphs) is unlikely to be achievable if we want to estimate all the eccentricities in the graph. 
Only a $5/3$ approximation is known for this problem in subquadratic time \cite{diametersoda14}.

\begin{reminder}{Theorem~\ref{thm:allecc}}
 A $(5/3-\delta)$ approximation algorithm for the eccentricities of all nodes in undirected sparse graphs that runs in subquadratic time refutes the Orthogonal Vectors Conjecture.
\end{reminder}

\begin{proof}
Let $A,B \subseteq \{0,1\}^d, |A|=|B|=n$ be an instance of OV, we will use it to construct an undirected sparse graph $G$ as follows.
Abusing the notation, construct a set of nodes $A$ that contains a node $a$ for every vector in $A$, and similarly construct a set of nodes $B$ from the vectors $B$.
Add a set of nodes $C$ corresponding to the coordinates $j \in [d]$ and for every vector $v \in A \cup B$ we add an edge $\{v,j\}$ iff $v[j]=0$.
Note that now, there is a $2$-path from a node $a \in A$ to a node $b\in B$ iff the vectors $a,b$ are not orthogonal.
Then, we also add a set $B'$ that contains a copy $b'$ of every node $b \in B$, such that $b'$ is only connected with one edge to $b$.
Finally, add two nodes $x,y$, connect $x$ to every node in $A$, connect $y$ to every node in $C$, and add the edge $\{x,y\}$.

We now claim that for every node $a\in A$, the eccentricity is $5$ if the vector $a$ is orthogonal to some vector $b \in B$, and it is $3$ otherwise.

To prove our claim we first show that any node $a \in A$ will have distance $\leq 3$ to all the nodes in $A \cup C \{x,y\}$:
there is a $2$-path via $x$ to every other node in $C$ and there is a $3$-path to every node in $C$ via $x,y$.
Now, on the one hand, if $a \in A$ is not orthogonal to any $b\in B$, then in our graph, there will be a $2$-path via $C$ from $a$ to every node in $B$ and therefore there will be a $3$-path to every node in $B'$.
Thus, in this case, the eccentricity is $3$.
On the other hand, if $a \in A$ is orthogonal to some $b \in B$, then the only paths from $a$ to $b$ have length $4$, and therefore the distance from $a$ to $b'$ is $5$, and the eccentricity of $a$ is at least $5$.

This claim shows that an algorithm estimating all the eccentricities within $(5/3-\delta)$ allows us to solve OV by this reduction.
It is important to note that the radius and diameter will not be determined by nodes in $A$ in our graph, and therefore a better than $5/3$ approximation for those parameter (which is achievable in subquadratic time \cite{RV13}) is not enough to solve OV.
Indeed, the node $x$ will have eccentricity $4$ regardless of the OV instance.

The number of nodes in the graph we constructed is $O(n)$ and the number of edges is $O(nd)$.
It can be easily turned into a sparse graph on $O(nd)$ nodes.
Thus, a subquadratic $(5/3-\delta)$ approximation allows us to solve OV in time $(n^{2-\eps} \cdot d^{2-\eps})$, for some $\eps>0$, which is enough to refute the OV conjecture.

\end{proof}

Our lower bounds for Diameter are summarized in Table~\ref{tab:res-diam}.

\begin{table*}\centering\small
\begin{tabular}{|c | c | c | c|}
\hline 
\multicolumn{4}{|c|}{Diameter Variants} \\
\hline
Problem & Definition & Upper Bound & OV Conjecture \\
\hline
  \UndirectedDiameter{} &
  $\max\limits_{u,v} d(u,v)$ &
  $3/2$ in $\tO(m\sqrt{n})$ [\cite{RV13}] &
  $3/2$  [\cite{RV13}]\\
\hline
  \MaxDiameter{} &
  $\max\limits_{u,v} d(u \to v)$ &
  $3/2$ in $\tO(m\sqrt{n})$ [\cite{RV13}] &
  $3/2$ [\cite{RV13}] \\
\hline
  \MinDiameter{} &
  $\max\limits_{u,v} \min\{d(u \to v), d(v \to u)\}$ &
  $n^\epsilon$ in $\tO(mn^{1-\epsilon})$ [Lem~\ref{lem:min-diameter-upper}] &
  $2$ on weighted  [Lem~\ref{lem:MinDiam}] \\
\hline
  \MinDiameter{} on DAGs &
  $\max\limits_{u<v} d(u \to v)$ &
  $2$ in $\tO(m)$ [Thm~\ref{thm:min-diameter-dag-upper}] &
  $3/2$  [Lem~\ref{lem:MinDiamDAG}] \\
\hline
  \RoundtripDiameter{} &
  $\max\limits_{u,v} \{d(u \to v) + d(v \to u)\}$ &
  $2$ in $\tO(m)$ [metric]&
  $3/2$ [Lem~\ref{lem:RoundDiam}] \\
\hline
\end{tabular}
\caption{Our Bounds for Various Diameter Problems}
\label{tab:res-diam}
\end{table*}